\documentclass{amsart}
\usepackage {amsmath, amssymb, color, textcomp}
\makeatletter

\newtheorem{thm}{Theorem}[section]
\newtheorem{cor}[thm]{Corollary}
\newtheorem{lem}[thm]{Lemma}
\newtheorem{defn}[thm]{Definition}
\newtheorem{prop}[thm]{Proposition}

\newcommand{\propref}[1]{Proposition~\ref{#1}}

\newcommand{\lemref}[1]{Lemma~\ref{#1}}
\newcommand{\eqnref}[1]{~(\ref{#1})}

\newtheorem{preremark}[thm]{Remark}
\newenvironment{remark}%
  {\begin{preremark}\upshape}{\end{preremark}}
\newtheorem{preexample}[thm]{Example}
  {\begin{preexample}\upshape}{\end{preexample}}

\numberwithin{equation}{section}

\numberwithin{equation}{section}

\numberwithin{thm}{section}
\newtheorem{lem/defn}[thm]{Lemma/Definition}
\newtheorem{preex/defn}[thm]{Example/Definition}
\newenvironment{ex/defn}%
  {\begin{preex/defn}\upshape}{\end{preex/defn}}

\newcommand{\ten}{\otimes}

\newcommand{\abs}[1]{\lvert#1\rvert}

\DeclareMathOperator{\End}{End}

\begin{document}

\title[Representations of $a_{\infty}$ and  $d_{\infty}$ with central charge 1 on the Fock space $\mathit{F^{\ten \frac{1}{2}}}$]{Representations of $a_{\infty}$ and  $d_{\infty}$  with central charge 1 on the single neutral fermion Fock space $\mathit{F^{\ten \frac{1}{2}}}$}
\author{Iana I. Anguelova, Ben Cox,  Elizabeth Jurisich}

\address{Department of Mathematics,  College of Charleston,
Charleston SC 29424, USA }
\email{anguelovai@cofc.edu, coxbl@cofc.edu, jurische@cofc.edu}


\begin{abstract}
We construct a new representation of the infinite rank Lie algebra $a_{\infty}$ with central charge $c=1$ on the Fock space $\mathit{F^{\ten \frac{1}{2}}}$ of a single neutral fermion. We show that  $\mathit{F^{\ten \frac{1}{2}}}$ is a direct sum of  irreducible integrable highest weight modules for $a_{\infty}$ with central charge $c=1$. We prove that as $a_{\infty}$ modules $\mathit{F^{\ten \frac{1}{2}}}$ is isomorphic to the Fock space $\mathit{F^{\ten 1}}$ of the charged free fermions. As a corollary we obtain the decompositions of  certain irreducible highest weight modules for $d_{\infty}$ with central charge $c=\frac{1}{2}$ into irreducible highest weight modules for $d_{\infty}$ with central charge $c=1$.
\end{abstract}

\maketitle

\section{Introduction}
\label{sec:intro}
Our motivation for this paper was to better understand the  various boson-fermion correspondences and their connection with the representation theory of  certain infinite dimensional Lie algebras. The first  to write on the topic of the relationship between a boson-fermion correspondence and representation theory were  Date, Jimbo, Kashiwara, Miwa in \cite{DJKM-1}, \cite{DJKM3}  and I. Frenkel in \cite{Frenkel-BF}. Since then many attempts have been made to understand the boson-fermion correspondences as nothing else but an isomorphism of infinite dimensional Lie algebra modules.
In his seminal paper I. Frenkel wrote: ``The Boson-Fermion correspondence is nothing else but the
\textbf{canonical} isomorphism between two realizations of the same representation
of the affine Lie algebra $\hat{D}(2l)$ and in particular of its subalgebra $\widehat{gl}(l)$" (page 317 of \cite{Frenkel-BF}). In the language that became commonly used later many had translated this to mean that the boson-fermion correspondence of type A is \textbf{just} an isomorphism between the vertex (the bosonic) and the spinor (the fermionic)  realizations of the standard modules of $a_{\infty}=\widehat{gl}_{\infty}$ (the label ``type A" derives from the $a_{\infty}=\widehat{gl}_{\infty}$, and is intended to distinguish this correspondence from the boson-fermion correspondence of type B for example, see e.g. \cite{DJKM6}, \cite{AngTVA}, \cite{Ang-Varna2}). This point of view of course had to be amended, as the charged free fermion Fock space $\mathit{F^{\ten 1}}$ underlying the fermionic side of the boson-fermion correspondence of type A  is actually an infinite  direct sum of irreducible  standard modules of $a_{\infty}=\widehat{gl}_{\infty}$ (for details on $\mathit{F^{\ten 1}}$  see e.g. \cite{Frenkel-BF}, \cite{KacRaina}, \cite{Kac}, \cite{WangDual}, as well as Remark \ref{remark:F^1} in this paper). Starting with I. Frenkel's work in \cite{Frenkel-BF}, and later, the boson-fermion correspondence of type A  (and the correspondence of type B) was related to different kinds of Howe-type dualities (\cite{WangKac}, \cite{WangDual}, \cite{WangDuality}). For instance, in \cite{WangDual} Wang wrote that  ``the $(GL_1, \hat{D})$ -duality in Theorem 5.3 is essentially the celebrated boson-fermion correspondence" ($\hat{D}$ denotes the universal central extension of the Lie algebra of differential operators on the circle, sometimes also labeled by $W_{1+\infty}$).
But what we contend is that the boson-fermion correspondences are more than just  isomorphisms between certain Lie algebra modules: a boson-fermion correspondence is first and foremost an isomorphism between two different   chiral field theories, one fermionic (expressible in terms of free fermions and their descendants), the other bosonic (expressible in terms  of exponentiated bosons). In fact,  as I. Frenkel was careful to summarize in Theorem II.4.1 of his very influential paper \cite{Frenkel-BF}, ``the canonical"
isomorphism of the two $o(2l)$-current algebra modules in the bosonic and the fermionic Fock spaces \textbf{follows} from the boson-fermion correspondence (in fact that is what makes the isomorphism canonical), but not vice versa.
But what we will show is that although the isomorphism of Lie algebra modules (and indeed the various dualities) follow
from a boson-fermion correspondence, the isomorphism as Lie algebra modules is  \textbf{not equivalent} to a
boson-fermion correspondence. To do that we consider a \textbf{single} neutral fermion Fock space  $\mathit{F^{\ten \frac{1}{2}}}$ and show that as modules for the Lie algebra $a_{\infty}$,  $\mathit{F^{\ten \frac{1}{2}}}\cong \mathit{F^{\ten 1}}$. Of course, it is known that even as super vertex algebras, and certainly as  modules for the Lie algebras $a_{\infty}$, $d_{\infty}$ with central charge $c=1$, as well as other affine Lie algebras,  $\mathit{F^{\ten 1}}\cong \mathit{F^{\ten \frac{1}{2}}}\ten \mathit{F^{\ten \frac{1}{2}}}$. This fact is often and extensively used in many papers on vertex algebras, and it was once again  I. Frenkel who used it first in \cite{Frenkel-BF} in connection to representation theory. But, the representations of  $a_{\infty}$, $d_{\infty}$ and other affine algebras on $\mathit{F^{\ten 1}}\cong \mathit{F^{\ten \frac{1}{2}}}\ten \mathit{F^{\ten \frac{1}{2}}}$  that are known in the literature do not restrict to representations on each of the $\mathit{F^{\ten \frac{1}{2}}}$ factors. What was known is that $\mathit{F^{\ten \frac{1}{2}}}$ is a representation of the Lie algebra $d_{\infty}$ with central charge $c=\frac{1}{2}$ (this is one of the explanations for the label $\frac{1}{2}$ in $\mathit{F^{\ten \frac{1}{2}}}$, the other being that $\mathit{F^{\ten \frac{1}{2}}}$ is only a ``half-infinite" Fock space, as opposed to $\mathit{F^{\ten 1}}$).
This is then what we do in this paper: First, we build a fermionic (spinor) representation of $a_{\infty}$ with central charge $c=1$ on $\mathit{F^{\ten \frac{1}{2}}}$. Next we show how this representation decomposes into irreducible highest weight modules, which ultimately shows that as modules for the Lie algebra $a_{\infty}$ with central charge $c=1$, $\mathit{F^{\ten \frac{1}{2}}}\cong \mathit{F^{\ten 1}}$.
This shows that it is not the $a_{\infty}$-module structure that distinguishes these spaces---$\mathit{F^{\ten \frac{1}{2}}}$ and $\mathit{F^{\ten 1}}$ are identical as vector spaces, or even as $a_{\infty}$ Lie algebra modules.  The difference is in the  vertex algebra  structure  (field theory) on $\mathit{F^{\ten \frac{1}{2}}}$, versus the vertex algebra structure on $\mathit{F^{\ten 1}}$.  The field theory on $\mathit{F^{\ten 1}}$ is local in the usual sense (at $z=w$, or as we can refer to it, 1-point local, see Definition \ref{defn:parity}); or more precisely $\mathit{F^{\ten 1}}$ has a \textbf{super vertex algebra} structure (see e.g. \cite{Kac}, \cite{LiLep}, \cite{FZvi} for a precise definition of a super vertex algebra). On the other hand, even though $\mathit{F^{\ten \frac{1}{2}}}$ has a super vertex algebra structure, this super vertex algebra structure is not enough to produce the new representations that we  obtain below-- to do that we  at the minimum  need to introduce 2-point locality (i.e., the fields we consider on $\mathit{F^{\ten \frac{1}{2}}}$ are allowed to be multi-local,  at both $z=w$ and $z=-w$). More precisely,  there is a \textbf{twisted vertex algebra} structure on $\mathit{F^{\ten \frac{1}{2}}}$ (see \cite{AngTVA}, \cite{ACJ} for a precise definition of a twisted vertex algebra). This shows that the type of vertex algebra structure on $\mathit{F^{\ten 1}}$ versus $\mathit{F^{\ten \frac{1}{2}}}$ is of great importance, in particular the set of points of locality is a necessary part of the data describing any boson-fermion correspondence.

The outlay  of the paper is as follows. First, we recall the necessary definitions and technical tools in Section \ref{section:background}. In Section \ref{section:main} we introduce the Fock space $\mathit{F^{\ten \frac{1}{2}}}$ and its different gradings, and the infinite rank Lie algebras that we will work with. Next we show that $\mathit{F^{\ten \frac{1}{2}}}$ is a module for the Lie algebra $a_{\infty}$ with central charge $c=1$, and by restriction for the Lie algebra $d_{\infty}$ with central charge $c=1$. We show that each homogeneous component of $\mathit{F^{\ten \frac{1}{2}}}$ is a highest weight module for $a_{\infty}$ with central charge $c=1$, which is moreover irreducible. That allows us to show that $\mathit{F^{\ten \frac{1}{2}}}$ is completely reducible and to obtain its decomposition in terms of irreducible modules for $a_{\infty}$ with central charge $c=1$. Hence we can compare and conclude that as  $a_{\infty}$ modules with central charge $c=1$  $\mathit{F^{\ten \frac{1}{2}}}\cong \mathit{F^{\ten 1}}$. Finally as a corollary we obtain the decomposition of  certain $c=\frac{1}{2}$ modules for $d_{\infty}$  in terms of irreducible highest weight  modules for $d_{\infty}$ with central charge $c=1$.

\section{Notation and background}
\label{section:background}

We work over the field of complex numbers $\mathbb{C}$.

The mathematical definitions of a field in a chiral quantum field theory and normal ordered products of fields are well known, they can be found for instance in \cite{FLM}, \cite{FHL},  \cite{Kac}, \cite{LiLep} and others,  we include them for completeness:
\begin{defn}
\begin{bf} (Field)\end{bf}\label{defn:field-fin}
 A field $a(z)$ on a vector space $V$ is a series of the form
\begin{equation}
a(z)=\sum_{n\in \mathbf{Z}}a_{(n)}z^{-n-1}, \ \ \ a_{(n)}\in
\End(V), \ \ \text{such that }\ a_{(n_v)}v=0 \ \ \text{for any}\ v\in V, \ n_v\gg 0.
\end{equation}
\end{defn}
Denote
\begin{equation}
a(z)_-:=\sum_{n\geq 0}a_nz^{-n-1},\quad a(z)_+:=\sum_{n<0}a_nz^{-n-1}.
\end{equation}
\begin{defn}(\cite{ACJ})  \label{defn:parity} {\bf ($N$-point local fields) }
Let $\epsilon$ be a primitive $N$th root of unity. We say that a field $a(z)$ on a vector space $V$ is {\bf even} and $N$-point self-local at $1, \epsilon, \epsilon^2, \dots, \epsilon^{N-1}$,  if there exist $n_0, n_1, \dots  , n_{N-1}\in \mathbb{Z}_{\geq 0}$ such that
\begin{equation}
(z- w)^{n_{0}}(z-\epsilon w)^{n_{1}}\cdots (z-\epsilon^{N-1} w)^{n_{N-1}}[a(z),a(w)] =0.
\end{equation}
In this case we set the {\bf parity} $p(a(z))$ of $a(z)$ to be $0$.
\\
We set $\{a, b\}:  =ab +ba$.We say that a field $a(z)$ on $V$ is $N$-point self-local at $1, \epsilon, \epsilon^2, \dots, \epsilon^{N-1}$
and {\bf odd} if there exist $n_0, n_1, \dots , n_{N-1}\in \mathbb{Z}_{\geq 0}$ such that
\begin{equation}
(z- w)^{n_{0}}(z-\epsilon w)^{n_{1}}\cdots (z-\epsilon^{N-1} w)^{n_{N-1}}\{a(z),a(w)\}=0.
\end{equation}
In this case we set the {\bf parity} $p(a(z))$ to be $1$. For brevity we will just write $p(a)$ instead of $p(a(z))$.\\
Finally,  if $a(z), b(z)$ are fields on $V$, we say that $a(z)$ and $b(z)$ are {\it $N$-point mutually local} at $1, \epsilon, \epsilon^2, \dots, \epsilon^{N-1}$
if there exist $n_0, n_1, \dots , n_{N-1} \in \mathbb{Z}_{\geq 0}$ such that
\begin{equation}
(z- w)^{n_{0}}(z-\epsilon w)^{n_{1}}\cdots (z-\epsilon^{N-1} w)^{n_{N-1}}\left(a(z)b(w)-(-1)^{p(a)p(b)}b(w)a(z)\right)=0.
\end{equation}
\end{defn}
\begin{defn} \label{defn:normalorder}{\bf (Normal ordered product)}
Let $a(z), b(z)$ be fields on a vector space $V$. Define
\begin{equation}
:a(z)b(w):=a(z)_+b(w)+(-1)^{p(a)p(b)}b(w)a_-(z).
\end{equation}
One calls this the ``normal ordered product" of $a(z)$ and $b(w)$.
\end{defn}
\begin{remark}
Let  $a(z), b(z)$ be any fields on a vector space $V$. Then \\
$:a(z)b(\lambda z):$ and $:a(\lambda z)b( z):$ are well defined fields on $V$ for any $\lambda \in \mathbb{C}^*$.
\end{remark}
For a rational function $f(z,w)$,  with poles only at $z=0$,  $z=\epsilon^i w, \ 0\leq i\leq N-1$, we denote by $i_{z,w}f(z,w)$
the expansion of $f(z,w)$ in the region $\abs{z}\gg \abs{w}$ (the region in the complex $z$ plane outside of all  the points $z=\epsilon^i w, \ 0\leq i\leq N-1$), and correspondingly for
$i_{w,z}f(z,w)$.
The mathematical background of the well-known and often used (both in physics and in mathematics) notion of Operator Product Expansion (OPE) of product of two fields for the case of usual locality ($N=1$) has been established for example in \cite{Kac}, \cite{LiLep}.
The following lemma extended the mathematical background  to the case of  $N$-point locality:
\begin{lem} (\cite{ACJ}) {\bf (Operator Product Expansion (OPE) of $N$-point local fields)}\label{lem:OPE}
 Let $a(z)$, $b(w)$ be {\it $N$-point mutually local} fields on a vector space $V$. Then exists fields $c_{jk}(w)$, $j=0, \dots, N-1; k=0, \dots , n_j-1$, such that we have
 \begin{equation}
 \label{eqn:OPEpolcor}
 a(z)b(w) =i_{z, w} \sum_{j=0}^{N-1}\sum_{k=0}^{n_j-1}\frac{c_{jk}(w)}{(z-\epsilon^j w)^{k+1}} + :a(z)b(w):.
 \end{equation}
We call the fields $c_{jk}(w)$, $j=0, \dots, N-1; k=0, \dots , n_j-1$ OPE coefficients. We will write the above OPE as
 \begin{equation}
 a(z)b(w) \sim  \sum_{j=0}^{N-1}\sum_{k=0}^{n_j-1}\frac{c_{jk}(w)}{(z-\epsilon_j w)^{k+1}}.
 \end{equation}
 \end{lem}
  The $\sim $ signifies that we have only written the singular part, and also we have omitted writing explicitly the expansion $i_{z, w}$, which we do acknowledge  tacitly. Often also the following notation is used for short:
 \begin{equation}\label{contraction}
\lfloor
ab\rfloor=a(z)b(w)-:a(z)b(w):= [a(z)_-,b(w)],
\end{equation}
i.e.,  the {\it contraction} of any two fields
$a(z)$ and $b(w)$ is in fact also the $i_{z, w}$ expansion of the singular part of the OPE of the two fields $a(z)$ and $b(w)$.

 The OPE expansion of the product of two fields is very convenient, as it completely determines in a very compact manner the commutation relations between the modes of the two fields, and we will use it extensively in what follows. In particular, extending of the OPEs to the case of N-point local fields allows us to extend and use Wick's Theorem for N-point local fields:
 \begin{thm}[Wick's Theorem, \cite{MR85g:81096}, \cite{MR99m:81001} or
\cite{Kac} ]  Let  $a^i(z)$ and $b^j(z)$ be {\it $N$-point mutually local} fields on a vector space $V$,
 satisfying
\begin{enumerate}
\item $[ \lfloor a^i(z)b^j(w)\rfloor ,c^k(x)_\pm]=[ \lfloor
a^ib^j\rfloor ,c^k(x)_\pm]=0$, for all $i,j,k$ and
$c^k(x)=a^k(z)$ or
$c^k(x)=b^k(w)$.
\item $[a^i(z)_\pm,b^j(w)_\pm]=0$ for all $i$ and $j$.
\end{enumerate}
Then
\begin{align*}
:&a^1(z)\cdots a^M(z)::b^1(w)\cdots
b^N(w):= \\
  &\sum_{s=0}^{\min(M,N)}\sum_{\substack{i_1<\cdots<i_s,\\
j_1\neq \cdots \neq j_s}}\pm \lfloor a^{i_1}b^{j_1}\rfloor\cdots
\lfloor a^{i_s}b^{j_s}\rfloor
:a^1(z)\cdots a^M(z)b^1(w)\cdots
b^N(w):_{(i_1,\dots,i_s;j_1,\dots,j_s)}.
\end{align*}
Here the subscript
${(i_1,\dots,i_s;j_1,\dots,j_s)}$ means that those factors $a^i(z)$,
$b^j(w)$ with indices
$i\in\{i_1,\dots,i_s\}$, $j\in\{j_1,\dots,j_s\}$ are to be omitted from
the product
$:a^1\cdots a^Mb^1\cdots b^N:$ and when $s=0$ we do not omit any factors.
The plus or minus sign is determined as follows:  each permutation of an adjacent odd field changes the sign.
\end{thm}

\section{The Fock space  $\mathit{F^{\ten \frac{1}{2}}}$ and representations of  $a_{\infty}$ and $d_{\infty}$ with central charge 1}
\label{section:main}

 We recall the definitions and notation for the Fock space $\mathit{F^{\ten \frac{1}{2}}}$ and the double-infinite rank Lie algebras $a_{\infty}$ and  $d_{\infty}$ as in \cite{Frenkel-BF}, \cite{DJKM3}, \cite{Kac-Lie}, \cite{WangDuality}; in particular we follow the notation of \cite{WangDuality}, \cite{WangDual}.

Consider a single odd self-local field $\phi ^D(z)$, which we index in the form $\phi ^D(z)=\sum _{n\in \mathbb{Z}+\frac{1}{2}} \phi^D_n z^{-n-\frac{1}{2}}$.
The OPE of $\phi ^D(z)$ is given by
\begin{equation}
\label{equation:OPE-D}
\phi ^D(z)\phi ^D(w)\sim \frac{1}{z-w}.
\end{equation}
This OPE completely determines the commutation relations between the modes $\phi^D_n$, $n\in \mathbb{Z} +\frac{1}{2}$:
\begin{equation}
\label{eqn:Com-D}
\{\phi^D_m,\phi^D_n\}=\phi^D_m\phi^D_n + \phi^D_n\phi^D_m=\delta _{m, -n}1.
\end{equation}
and so the modes generate a Clifford algebra $\mathit{Cl_D}$. The field $\phi ^D(z)$ is usually  called a ``neutral fermion field".  Now $Cl_D$ has basis consisting of $1$ and  the products $\phi_{i_1}^D\phi_{i_2}^D \cdots\phi_{i_k}^D$ where $i_1< i_2<  \cdots< i_k$, $i_j \in \mathbb Z +1/2$. We introduce a $\mathbb Z$-grading $dg$ on $Cl_D$ by defining the following degree of a basis element:$$dg(1) = 0,$$
\begin{equation*}
 dg(\phi^D_{n_k-\frac{1}{2}}\dots \phi^D_{n_2-\frac{1}{2}}\phi^D_{ n_1-\frac{1}{2}})
 =\#\{i=1, 2, \dots, k|\  n_i=\text{odd}\}\,  -\#\{i=1, 2, \dots, k|\  n_i=\text{even}\}.
 \end{equation*}
\begin{lem}\label{grading}
The $\mathbb Z$-grading of $Cl_D$ is an algebra grading. Furthermore, the operation left multiplication by an element $\phi_n^D\in CL_D$ for any $n\in \mathbb Z +1/2$ is a homogenous operator on $CL_D$, of degree $1$ if $n = 2k+1/2$, and of degree $-1$  if $n = 2k - 1/2$  for some integer $k$.
 \end{lem}
\begin{proof} For any pair $\phi_m, \phi_n$ ($m\neq n$) we claim $dg (\phi_m \phi_n) = dg (\phi_n \phi_m)$. If $m \neq -n$ then $\phi_m \phi_n = -\phi_n \phi_m$, and one of these expressions appears in the given basis so determines the degree of both $\phi_m \phi_n$ and $ \phi_n \phi_m$. If $m = -n$, $\phi_m \phi_n = 1 -\phi_n \phi_m$, and again, either the left or right hand side is a sum of basis vectors, since $dg (1) = 0$ the degree of $\phi_m \phi_n$ and $ \phi_n \phi_m$ again agree. Thus the Clifford algebra relation \eqref{eqn:Com-D}   is compatible with the definition of $dg$.
  Now it is obvious from the definition of the grading that as an operator left multiplication by $\phi_n^D\in CL_D$ is a homogeneous operator of the given degree.
  \end{proof}
The Fock space of the  field $\phi ^D(z)$ is the highest weight module of $\mathit{Cl_D}$ with vacuum vector $|0\rangle $, so that $\phi^D_n|0\rangle=0 \ \text{for} \  n >0$.  It is denoted  by $\mathit{F^{\ten \frac{1}{2}}}$ (see e.g.  \cite{DJKM6}, \cite{Triality}, \cite{Wang},  \cite{WangDual}, \cite{WangDuality}, \cite{WangKac}).
$\mathit{F^{\ten \frac{1}{2}}}$  has basis
\begin{equation}
\{ \phi^D_{-n_k-\frac{1}{2}}\dots \phi^D_{-n_2-\frac{1}{2}}\phi^D_{-n_1-\frac{1}{2}}|0\rangle, |0\rangle \big |\ n_k>\dots >n_2>n_1\geq 0, n_i\in \mathbb{Z}, i=1, 2, \dots, k\}
\end{equation}
The space $\mathit{F^{\ten \frac{1}{2}}}$ has  a $\mathbb{Z}_2$ grading given by  $ k\  mod\  2$,
\[
\mathit{F^{\ten \frac{1}{2}}}=\mathit{F_{\bar{0}}^{\ten \frac{1}{2}}}\oplus \mathit{F_{\bar{1}}^{\ten \frac{1}{2}}},
\]
where $\mathit{F_{\bar{0}}^{\ten \frac{1}{2}}}$ (resp. $\mathit{F_{\bar{1}}^{\ten \frac{1}{2}}}$) denote the even (resp. odd) components of $\mathit{F^{\ten \frac{1}{2}}}$. This $\mathbb{Z}_2$ grading can be extended  to a $\mathbb{Z}_{\geq 0}$ grading $\tilde{L}$, called ``length", by setting
\begin{equation}
\tilde{L} (\phi^D_{-n_k-\frac{1}{2}}\dots \phi^D_{-n_2-\frac{1}{2}}\phi^D_{-n_1-\frac{1}{2}}|0\rangle)=k.
\end{equation}
The space $\mathit{F^{\ten \frac{1}{2}}}$ can be given a super vertex algebra structure, as is known from e.g. \cite{Triality}, \cite{Wang}, \cite{Kac}.

The $\mathbb Z$ grading $dg$ on $Cl_D$ induces a $\mathbb{Z}$ grading $dg$ on $\mathit{F^{\ten \frac{1}{2}}}$ by assigning $dg(|0\rangle)=0$ and
\begin{align}\label{grading2}
dg(\phi^D_{-n_k-\frac{1}{2}}\dots \phi^D_{-n_2-\frac{1}{2}}\phi^D_{-n_1-\frac{1}{2}}|0\rangle)&=\#\{i=1, 2, \dots, k| \ n_i=\text{odd}\}\\&\quad -\#\{i=1, 2, \dots, k| \ n_i=\text{even}\}.\notag
\end{align}
Denote the  space of homogenous elements  of degree $dg=n \in \mathbb Z$ by $\mathit{F_{(n)}^{\ten \frac{1}{2}}}$, hence as vector spaces we have
\begin{equation}
\mathit{F^{\ten \frac{1}{2}}} = \oplus _{n\in \mathbb{Z}} \mathit{F_{(n)}^{\ten \frac{1}{2}}}.
\end{equation}
Introduce  also the special vectors $v_n\in  \mathit{F_{(n)}^{\ten \frac{1}{2}}}$ defined by
\begin{align}
&v_0=|0\rangle \in  \mathit{F_{(0)}^{\ten \frac{1}{2}}};\\
&v_n=\phi^D_{-2n+1-\frac{1}{2}}\dots \phi^D_{-3-\frac{1}{2}}\phi^D_{-1-\frac{1}{2}}|0\rangle \in  \mathit{F_{(n)}^{\ten \frac{1}{2}}}, \quad \text{for}\ n>0;\\
&v_{-n}=\phi^D_{-2n+2-\frac{1}{2}}\dots \phi^D_{-2-\frac{1}{2}}\phi^D_{-\frac{1}{2}}|0\rangle\in  \mathit{F_{(-n)}^{\ten \frac{1}{2}}}, \quad \text{for}\ n>0.
\end{align}
Note that the vectors $v_n\in \mathit{F_{(n)}^{\ten \frac{1}{2}}}$ have minimal length $\tilde{L}=|n|$  among the vectors within $\mathit{F_{(n)}^{\ten \frac{1}{2}}}$, and they are in fact the unique (up-to a scalar) vectors minimizing  the length $\tilde{L}$, such that the index $n_k$ is minimal too.

 The  Lie algebra  $\bar{a}_{\infty}$ (sometimes denoted $\bar{gl}_{\infty}$ or just $\mathfrak {gl}$, see for instance \cite{WangDual}, \cite{WangDuality}, \cite{WangKac}) is the Lie algebra of infinite matrices
\begin{equation}
\bar{a}_{\infty}=\{(a_{ij}) | \ i, j\in \mathbb{Z}, \ a_{ij} =0 \ \text{for} |i-j|\gg 0 \}.
\end{equation}
As usual denote the elementary matrices by $E_{ij}$.

The algebra $a_{\infty}$ (often  denoted also by $\widehat{gl}_{\infty}$ or $\widehat{\mathfrak {gl}}$) is a central extension of $\bar{a}_{\infty}$ by a central element $c$, $a_{\infty}=\bar{a}_{\infty}\oplus \mathbb{C}c$, with cocycle given by
\begin{equation}
\label{equation:cocycle-a}
 C(A, B) =Trace([J, A]B),
\end{equation}
where the matrix $J=\sum_{i\le 0}E_{ii}$. In particular
\begin{align*}
C(E_{ij},E_{ji})&=-C(E_{ji},E_{ij})=1,\quad \text{if} \enspace i\leq 0,\enspace j\geq 1 \\
C(E_{ij},E_{kl})&=0\quad \text{ in all other cases}.
\end{align*}

The commutation relations for the elementary matrices in $a_\infty$ are
\begin{align*}
[E_{ij},E_{kl}]=\delta_{jk}E_{il}-\delta_{li}E_{kj}+ C(E_{ij},E_{kl})c.
\end{align*}
The non-central generators have generating series
\begin{equation}
E^A(z, w) =\sum _{i, j\in \mathbb{Z}} E_{ij}z^{i-1}w^{-j},
\end{equation}
and relations
\begin{align}
[E^A(z_1, w_1)&,  E^A (z_2, w_2)] =E^A(z_1, w_2)\delta(z_2-w_1) -E^A(z_2, w_1)\delta(z_1-w_2) \\
&\quad +\iota_{z_1,w_2}\frac{1}{z_1-w_2}\iota_{w_1, z_2}\frac{1}{w_1-z_2}c -\iota_{w_2, z_1}\frac{1}{z_1-w_2}\iota_{ z_2, w_1}\frac{1}{w_1-z_2}c.\notag
 \end{align}
Here we used the formal delta function notation $\delta(z-w): =\sum_{n\in \mathbb{Z}}z^nw^{-n-1}=\delta (w-z)$ (see e.g. \cite{Kac}, \cite{FZvi}, \cite{ACJ}).

Further,  $a_{\infty}$  has a triangular decomposition
\begin{equation}
 a_\infty= a^-_\infty\oplus a_\infty^0\oplus a^+_\infty.
\end{equation}
Here
$a^{\pm}_\infty$ consists of  correspondingly the strictly upper (strictly lower) triangular infinite matrices;   $a_\infty^0=\mathfrak {gl}_0\oplus \mathbb Cc$ where $\mathfrak {gl}_0$ denotes the diagonal matrices.

The root system of $a_{\infty}$ is $\Delta=\{\epsilon_i-\epsilon_j\,|\,i,j\in\mathbb Z,i\neq j\}$ where $\epsilon_i\in (\mathfrak{gl}_0)^*$ is defined by $\epsilon_i(E_{jj})=\delta_{ij}$ ($i,j\in\mathbb Z)$. There is a conjugate linear, involutive anti-automorphism $\omega\in \text{End}( a_\infty)$ defined by $\omega(E_{ij})=E_{ji}$ and this is called ``the compact anti-involution".

For \textcent$\in \mathbb C$ and $\Lambda\in\bigoplus_{i\in\mathbb Z,i\neq 0}(\mathbb CE_{ii})^*$, set
\begin{align*}
^a\lambda_i&:=\Lambda(E_{ii})  \\
^aH_i&:=E_{ii}-E_{i+1,i+1}+\delta_{i,0}c  \\
^ah_i&:=\Lambda({^aH_i})={^a}\lambda_i-{^a}\lambda_{i+1}+\delta_{i,0}\text{\textcent}
\end{align*}
Define $^a\!\Lambda_j\in (a_\infty^0)^*$ by
\begin{align*}
^a\!\Lambda_j(E_{ii})&=\begin{cases} 1, &\text{ for }0<i \leq j,\\
-1, &\text{ for }j<i\leq 0,\\
0,&\text{ otherwise,}
\end{cases} \\
^a\!\Lambda_j(c)&=0.
\end{align*}
Define also $^a\!\hat\Lambda_0\in (a_\infty^{0})^*$ by
$
^a\!\hat\Lambda_0(c)=1$, $ ^a\!\hat\Lambda_0(E_{ii})=0$ for $i\in\mathbb Z$.  Then the $i$-th fundamental weight is
$$
^a\!\hat\Lambda_j= ^a\!\Lambda_j+^a\!\hat\Lambda_0,\quad i\in\mathbb Z.
$$
Let $L( a_\infty; ^a\!\Lambda,\text{\textcent})=L(\widehat{\mathfrak{gl}}_\infty; ^a\!\Lambda,\text{\textcent})$ denote the highest weight $ a_\infty$-module with highest weight $\Lambda$ and central charge \textcent.

The algebra   $\bar{d}_{\infty}$ is  defined as the \textbf{subalgebra} of $\bar{a}_{\infty}$,  consisting of the infinite matrices preserving the bilinear form $D(v_i, v_j)=\delta_{i, 1-j}$, i.e.,
\begin{equation}
\bar{d}_{\infty}=\{(a_{ij})\in \bar{a}_{\infty} | \ a_{ij}=-a_{1-j, 1-i} \}.
\end{equation}
Denote by $d_{\infty}$ the  central extension of $\bar{d}_{\infty}$ by a central element $c$, $d_{\infty}=\bar{d}_{\infty}\oplus \mathbb{C} c$, with the same cocycle as for $a_{\infty}$, \eqref{equation:cocycle-a}.  The commutation relations for the elementary matrices in $d_\infty$ are obtained using the relations in $a_\infty$: \footnote{Note that in \cite{Kac-Lie} the commutation relation  $[E_{ij},E_{kl}]=\delta_{jk}E_{il}-\delta_{li}E_{kj}+\mathbf{\frac{1}{2}}C(E_{ij},E_{kl})c$ is used instead.}
\begin{align*}
[E_{ij},E_{kl}]=\delta_{jk}E_{il}-\delta_{li}E_{kj}+ C(E_{ij},E_{kl})c.
\end{align*}
The generators for the algebra $d_\infty$ can be written in terms of these elementary matrices as:
\[
\{ E_{i, j} - E_{1-j, 1-i}, \  i, j \in \mathbb{Z}; \text{and} \ \ c\}.
\]
We can arrange the non-central generators in a generating series
\begin{equation}
E^D(z, w) =\sum _{i, j\in \mathbb{Z}} (E_{ij}-E_{1-j, 1-i})z^{i-1}w^{-j}.
\end{equation}
The generating series  $E^D(z,w)$ obeys the following relations:
\[
E^D(z,w) = -E^D(w,z)
\]
and
\begin{align*}
[E^D(z_1, w_1), &E^D(z_2, w_2)]
 =  E^D(z_1, w_2)\delta(z_2-w_1) -E^D(z_2, w_1)\delta(z_1-w_2) \\
&\hspace{1.7cm} +E^D(w_2,w_1)\delta(z_1-z_2) -E^D(z_1,z_2)\delta(w_1-w_2)  \\
&\quad +2\iota_{z_1,w_2}\frac{1}{z_1-w_2}\iota_{w_1,z_2}\frac{1}{w_1-z_2}c   -2\iota_{z_2,w_1}\frac{1}{z_2-w_1}\iota_{w_2,z_1}\frac{1}{w_2-z_1}c  \\
&\quad -2\iota_{z_1,z_2}\frac{1}{z_1-z_2}\iota_{w_1,w_2}\frac{1}{w_1-w_2}c +2\iota_{z_2,z_1}\frac{1}{z_1-z_2}\iota_{w_2,w_1}\frac{1}{w_2-w_1}c.
\end{align*}
The assignment $E^D(z, w)\mapsto :\phi^D(z)\phi^D (w):$, \ $c\mapsto \frac{1}{2}Id_{\mathit{F^{\ten \frac{1}{2}}}}$ gives a representation of the Lie algebra $d_{\infty}$ on $ \mathit{F^{\ten \frac{1}{2}}}$ (see e.g. \cite{DJKM6}, \cite{Wang}, \cite{WangKac}, \cite{ACJ}), which we denote by $r_{\frac{1}{2}}$.
Further, it is known (see e.g. \cite{Wang}, \cite{WangKac}, \cite{WangDual}) that as $d_{\infty}$ modules
\begin{equation}
\mathit{F_{\bar{0}}^{\ten \frac{1}{2}}}\cong L( d_\infty;^d\!\hat \Lambda_0,\frac{1}{2}); \quad \mathit{F_{\bar{1}}^{\ten \frac{1}{2}}}\cong L( d_\infty;^d\!\hat \Lambda_1, \frac{1}{2}).
\end{equation}
where
$L( d_\infty;^d\!\Lambda,\text{\textcent})$ denotes the highest weight $d_\infty$-module with highest weight $^d\!\Lambda$ and central charge \textcent. The highest weights above are defined by the following, using the symmetry in $d_\infty$:
\begin{align*}
& ^d\!\hat\Lambda_0(E_{i,i}- E_{1-i,1-i})=0\quad \text{for any }\ i\in \mathbb{Z};\\
& ^d\!\hat\Lambda_1(E_{i,i}- E_{1-i,1-i})=1; \quad \text{for} \ i=1;\quad  ^d\!\hat\Lambda_1(E_{i,i}- E_{1-i,1-i})=0\quad \text{for}\ i\neq 0, 1;\\
& ^d\!\hat\Lambda_0(c)=^d\!\hat\Lambda_1(c)=\frac{1}{2}.
\end{align*}
As a  $d_\infty$-module  with central charge $c=\frac{1}{2}$ $\mathit{F^{\ten \frac{1}{2}}}$ then decomposes as
\[
\mathit{F^{\ten \frac{1}{2}}}=\mathit{F_{\bar{0}}^{\ten \frac{1}{2}}}\oplus \mathit{F_{\bar{1}}^{\ten \frac{1}{2}}}\cong L( d_\infty;^d\!\hat \Lambda_0, \frac{1}{2})\oplus L( d_\infty;^d\!\hat \Lambda_1, \frac{1}{2}),
\]
Next we will show that $\mathit{F^{\ten \frac{1}{2}}}$  is also a module for $a_{\infty}$ (and thus $d_{\infty}$) with central charge $c=1$.
\begin{remark}\label{remark:F^1}
It is well known (in the context of representation theory it was introduced  by I. Frenkel in \cite{Frenkel-BF} and extensively used afterwards) that
\[
\mathit{F^{\ten \frac{1}{2}}}\ten \mathit{F^{\ten \frac{1}{2}}} \cong \mathit{F^{\ten 1}};
\]
where $\mathit{F^{\ten 1}}$ is the Fock space of \textbf{1 pair} of two charged fermions. The two charged fermions are  the fields $\psi^+ (z)$ and $\psi^- (z)$ with  operator product expansions (OPEs):
\begin{align*}
\psi^+ (z)\psi^- (w)\sim \frac{1}{z-w}\sim \psi ^-(z)\psi ^+(w),\quad \psi^+ (z)\psi^+ (w)\sim 0\sim \psi ^-(z)\psi ^-(w)
\end{align*}
where the $1$ above denotes the identity map $Id_{\mathit{F^{\ten 1}}}$.
We index the fields  $\psi^+ (z)$ and $\psi^- (z)$ as follows:
\begin{equation}
\psi ^+(z) =\sum _{n\in \mathbb{Z}} \psi^+_{n} z^{-n-1}, \quad \psi^- (z) =\sum _{n\in \mathbb{Z}} \psi^-_n z^{-n-1};
\end{equation}
and their modes $\psi^+_n$ and $\psi^-_n$, $n\in \mathbb{Z}$ form a Clifford algebra $\mathit{Cl_A}$  with relations
\begin{equation}
\{\psi^+_m,\psi^-_n\}=\delta _{m+n, -1}1, \quad \{\psi^+_m,\psi^+_n\}=\{\psi^-_m,\psi^-_n\}=0.
\end{equation}
The Fock space $\mathit{F^{\ten 1}}$ is the highest weight   representation of $\mathit{Cl_A}$ generated by the vacuum  vector $|0\rangle $,  so that $\psi^+_n|0\rangle=\psi^-_n|0\rangle=0 \ \text{for} \  n\geq 0$  (see e.g. \cite{Frenkel-BF},  \cite{KacRaina}, \cite{Wang}, \cite{Kac} for more details on $\mathit{F^{\ten 1}}$).
It is well known (see e.g., \cite{Kac}, \cite{FZvi}, \cite{LiLep}) that $\mathit{F^{\ten 1}}$ has a structure of a super vertex algebra (i.e., with a single point of locality at $z=w$ in the OPEs); this vertex algebra  is often called  ``charged free fermion super vertex algebra". It is also well known (introduced by I. Frenkel, \cite{Frenkel-BF}; and \cite{DJKM3}) and extensively used (e.g., \cite{KacRaina}, \cite{Triality}, \cite{Kac-Lie}, \cite{WangDuality}, \cite{WangDual} among many others) that  $\mathit{F^{\ten 1}}$ is a module for the $a_{\infty}$ algebra, moreover
\[
\mathit{F^{\ten 1}}\cong \oplus_{n\in \mathbb{Z}} L(a_\infty; ^a\!\hat\Lambda_n, 1).
\]
This isomorphism has often been referred to as equivalent to  the well known charged free boson-fermion correspondence (an isomorphism between the super vertex algebra on $\mathit{F^{\ten 1}}$ and the super vertex algebra of the rank one odd lattice). Here we are concerned with a more subtle point:   a boson-fermion correspondence is not just an isomorphism of Lie algebra modules, but an isomorphism of  appropriate field theories (or vertex algebra structures).  As shown and used by I.  Frenkel  (\cite{Frenkel-BF}) one has
\[
\mathit{F^{\ten 1}}\cong \mathit{F^{\ten \frac{1}{2}}}\ten \mathit{F^{\ten \frac{1}{2}}}  ,
\]
which is implemented by
\[
 \phi_1(z)=\frac{1}{\sqrt{2}}\left(\psi ^+(z)+\psi ^-(z)\right), \ \phi_2(z)=\frac{i}{\sqrt{2}}\left(\psi ^+(z)-\psi ^-(z)\right).
\]
In today's language we would say that this   mapping of the fields generates an isomorphism of super vertex algebras, a fact which is extensively used by many authors (e.g.,  \cite{Triality}, \cite{MasonTuiteZ}, \cite{Katrina} among many).
But as we will show below,  as $a_{\infty}$ modules
\[
\mathit{F^{\ten 1}}\cong  \mathit{F^{\ten \frac{1}{2}}} \cong \oplus_{n\in \mathbb{Z}} L(a_\infty; ^a\!\hat\Lambda_n, 1).
\]
In other words, as vector spaces and as $a_{\infty}$ modules $\mathit{F^{\ten \frac{1}{2}}}$ and $\mathit{F^{\ten 1}}$ are isomorphic, even though as super vertex algebras $\mathit{F^{\ten \frac{1}{2}}}\ten \mathit{F^{\ten \frac{1}{2}}}$ and $\mathit{F^{\ten 1}}$ are isomorphic.  The difference comes from the vertex algebra structures: in the case of $\mathit{F^{\ten 1}}$ the boson-fermion correspondence (of type A) is an isomorphism of \textbf{super vertex algebras}, requiring locality only at $z=w$ (see e.g., \cite{Kac}). But, in the  case of $\mathit{F^{\ten \frac{1}{2}}}$ the boson-fermion correspondence (of type D-A) is an isomorphism of \textbf{twisted vertex algebras}, requiring locality at both $z=w$ and $z=-w$ (see \cite{AngTVA}, \cite{ACJ}).
\end{remark}

 Now we proceed to obtain a representation of $a_{\infty}$ on $\mathit{F^{\ten \frac{1}{2}}}$ by introducing fields arising from a 2-point local twisted vertex algebra:
\begin{prop}\label{prop:r_1}
Let
\begin{equation}
\phi^{+DA}(z)=\frac{\phi^D(z)-\phi^D(-z)}{2}, \quad \phi^{-DA}(z)=\frac{\phi^D(z)+\phi^D(-z)}{2}.
\end{equation}
The assignment $zE(z^2, w^2)\mapsto :\phi^{+DA}(z)\phi^{-DA} (w):$, \ $c\mapsto Id_{\mathit{F^{\ten \frac{1}{2}}}}$ gives a representation of the Lie algebra $a_{\infty}$ on $\mathit{F^{\ten \frac{1}{2}}}$ with central charge $c=1$.
\end{prop}
\begin{proof}
We will use Wick's Theorem. We have  $\phi^{+DA}(z)=\sum _{n\in \mathbb{Z}} \phi^D _{-2n+\frac{1}{2}} z^{2n-1} $ and
$\phi^{-DA}(z)=\sum _{n\in \mathbb{Z}} \phi^D _{-2n-\frac{1}{2}} z^{2n} $, thus the modes  $\phi^{+DA}(z)$ (resp. $\phi^{-DA}(z)$) are the operator coefficients of  $\phi^{D}(z)$ in front of odd (resp. even) powers of the formal variable $z$. Hence, since  $\phi^{D}(z)$ obeys the second condition  of Wick's theorem, $\phi^{+DA}(z)$ and $\phi^{-DA}(z)$ obey it too.  We also have
\begin{align}
\phi^{+DA}(z)\phi^{+DA}(w)&\sim 0,  \quad \phi^{-DA}(z)\phi^{-DA}(w)\sim 0; \\
\phi^{+DA}(z)\phi^{-DA}(w)&\sim  \frac{1}{2}  \left(\frac{1}{z-w}+\frac{1}{z+w}\right)\sim  \frac{z}{z^2-w^2};\\
\phi^{-DA}(z)\phi^{+DA}(w)&\sim \frac{1}{2} \left(\frac{1}{z-w}-\frac{1}{z+w}\right)\sim   \frac{w}{z^2-w^2};
\end{align}
hence the first condition of Wick's theorem is also satisfied. Thus from Wick's theorem we have
\begin{align*}
:\phi^{+DA}(z_1)&\phi^{-DA}(w_1): :\phi^{+DA}(z_2)\phi^{-DA}(w_2): \\
& \sim \frac{z_1}{z_1^2-w_2^2}:\phi^{-DA}(w_1)\phi^{+DA}(z_2): +  \frac{z_2}{w_1^2-z_2^2}:\phi^{+DA}(z_1)\phi^{-DA}(w_2):\\
&\quad +\frac{z_1}{z_1^2-w_2^2}\frac{z_2}{w_1^2-z_2^2}.
\end{align*}
Hence
\begin{align*}
&[:\phi^{+DA}(z_1)\phi^{-DA}(w_1):,  :\phi^{+DA}(z_2)\phi^{-DA}(w_2):] \\
& = z_1\delta(z_1^2-w_2^2):\phi^{-DA}(w_1)\phi^{+DA}(z_2): + z_2\delta(z_2^2-w_1^2):\phi^{+DA}(z_1)\phi^{-DA}(w_2):\\
&\quad +\iota_{z_1,w_2}\frac{z_1}{z_1^2-w_2^2}\iota_{w_1, z_2}\frac{z_2}{w_1^2-z_2^2} -\iota_{w_2, z_1}\frac{z_1}{z_1^2-w_2^2}\iota_{z_2, w_1}\frac{z_2}{w_1^2-z_2^2};
\end{align*}
and we use the fact that $:\phi^{-DA}(w)\phi^{+DA}(z):=-:\phi^{+DA}(z)\phi^{-DA}(w):$.
\end{proof}
We will denote this new representation on $\mathit{F^{\ten \frac{1}{2}}}$ by $r_1$.
Since $d_{\infty}$ is a subalgebra of $a_{\infty}$, we have the following
\begin{cor}
$\mathit{F^{\ten \frac{1}{2}}}$ is a module for $d_{\infty}$  with central charge $c=1$ via the restriction of the representation $r_1$.
\end{cor}
If we introduce a normal ordered product $:\phi^D_m \phi^D_n:$ on the modes $\phi^D_m$ of the field $\phi^{D}(z)$, compatible with the normal ordered product of fields (Definition \ref{defn:normalorder}), we have to have
\[
:\phi^{D}(z)\phi^{D}(w): =\sum _{m,n\in \mathbf{Z+\frac{1}{2}}} :\phi^D_{-m-\frac{1}{2}} \phi^D_{-n-\frac{1}{2}}:z^{m}w^{n},
\]
and thus
\begin{align}
:\phi^D_{-m-\frac{1}{2}}& \phi^D_{-n-\frac{1}{2}}:= \begin{cases} \phi^D_{-m-\frac{1}{2}}\phi^D_{-n-\frac{1}{2}} &\quad \text{for}\  m+n\neq 1\\
  \phi^D_{-m-\frac{1}{2}}\phi^D_{-n-\frac{1}{2}} -1=-\phi^D_{-n-\frac{1}{2}}\phi^D_{-m-\frac{1}{2}} &\quad \text{for}\ m+n= -1, n\geq 0,\\
  \phi^D_{-m-\frac{1}{2}}\phi^D_{-n-\frac{1}{2}} & \quad \text{for}\ m+n= -1, m\geq 0.\end{cases}
\end{align}
Hence the well known representation $r_{\frac{1}{2}}$ of $d_{\infty}$ on $\mathit{F^{\ten \frac{1}{2}}}$ with central charge $c=\frac{1}{2}$ is defined by
\[
r_{\frac{1}{2}}(E_{m, n} - E_{1-n, 1-m})=:\phi^D_{-m +1/2} \phi^D_{n-\frac{1}{2}}:
\]
for all $m,n\in\mathbb Z$. Now Proposition \ref{prop:r_1} gives us a new representation $r_1$ of   $a_{\infty}$ on $\mathit{F^{\ten \frac{1}{2}}}$  given  by
\begin{equation}\label{ainftyaction}
r_1(E_{m, n})  =:\phi^D_{-2m+\frac{1}{2}} \phi^D_{2n-\frac{1}{2}}:.
\end{equation}
for all $m,n \in\mathbb Z$.
Hence by restriction the representation of $d_{\infty}$ on $\mathit{F^{\ten \frac{1}{2}}}$ with central charge $c=1$ is
\[
r_1(E_{m, n} - E_{1-n, 1-m}):=:\phi^D_{-2m+\frac{1}{2}} \phi^D_{2n-\frac{1}{2}}:- :\phi^D_{2n-3/2} \phi^D_{-2m+3/2}:.
\]
for all $m,n\in\mathbb Z$.
\begin{prop} \textbf(The upper triangular elements  annihilate  $v_n$) \label{singularvector}
 For all $n\in\mathbb Z$, $
a_\infty^+v_n=0.$
\end{prop}
\begin{proof} We need to prove that for any $k\geq 1$ and any $i, n\in \mathbb{Z}$, \ $r_1(E_{i, i+k})v_n=0$.
We have $r_1(E_{i, i+k})=:\phi^D_{-2i+\frac{1}{2}} \phi^D_{2i+2k-\frac{1}{2}}:$.
We start with $n=0$, $v_0=|0\rangle$. There are two cases:
The case of $2i+2k> 0$ is  trivial.  In the case  $2i+2k\leq 0$, then
$:\phi^D_{-2i+\frac{1}{2}} \phi^D_{2i+2k-\frac{1}{2}}: =-\phi^D_{2i+2k-\frac{1}{2}}\phi^D_{-2i+\frac{1}{2}} $. But $-2i\geq 2k\geq 2$, hence
\begin{equation*}
r_1(E_{i, i+k})v_0= r_1(E_{i, i+k})|0\rangle= -\phi^D_{2i+2k-\frac{1}{2}}\phi^D_{-2i+\frac{1}{2}}|0\rangle=0.
\end{equation*}
We continue  with $n>0$, where
$v_n=\phi^D_{-2n+1-\frac{1}{2}}\dots \phi^D_{-3-\frac{1}{2}}\phi^D_{-1-\frac{1}{2}}|0\rangle$.
  Again, there are two cases: first we will consider the case when $2i+2k\leq 0$:
\begin{equation*}
r_1(E_{i, i+k}) \phi^D_{-2n+1-\frac{1}{2}}\dots \phi^D_{-3-\frac{1}{2}}\phi^D_{-1-\frac{1}{2}}|0\rangle
=-\phi^D_{2i+2k-\frac{1}{2}}\phi^D_{-2i+\frac{1}{2}} \phi^D_{-2n+1-\frac{1}{2}}\dots \phi^D_{-3-\frac{1}{2}}\phi^D_{-1-\frac{1}{2}}|0\rangle.
\end{equation*}
Now since it is impossible to have $-2i + \frac{1}{2}=-(-2l +1-\frac{1}{2})$ for any $l\in \mathbb{Z}$, $\phi^D_{-2i+\frac{1}{2}}$ will anticommute with any $\phi^D_{-2l+1-\frac{1}{2}}$, and thus
$\phi^D_{-2i+\frac{1}{2}} \phi^D_{-2n+1-\frac{1}{2}}\dots \phi^D_{-3-\frac{1}{2}}\phi^D_{-1-\frac{1}{2}}|0\rangle=0$.

Next let $2i+2k> 0$, then $:\phi^D_{-2i+\frac{1}{2}} \phi^D_{2i+2k-\frac{1}{2}}: =\phi^D_{-2i+\frac{1}{2}} \phi^D_{2i+2k-\frac{1}{2}}$, and
\begin{equation*}
r_1(E_{i, i+k}) \phi^D_{-2n+1-\frac{1}{2}}\dots \phi^D_{-3-\frac{1}{2}}\phi^D_{-1-\frac{1}{2}}|0\rangle
=\phi^D_{-2i+\frac{1}{2}}\phi^D_{2i+2k-\frac{1}{2}} \phi^D_{-2n+1-\frac{1}{2}}\dots \phi^D_{-3-\frac{1}{2}}\phi^D_{-1-\frac{1}{2}}|0\rangle.
\end{equation*}
Now  $\phi^D_{2i+2k-\frac{1}{2}}$ is an operator annihilating the vacuum, and unless we have $2i+2k- \frac{1}{2}=-(-2l+1 -\frac{1}{2})$ for some $1\leq l\leq n, \ l\in \mathbb{Z}$, then it will anticommute with any of the $\phi^D_{-2l+1-\frac{1}{2}}$, $1\leq l\leq n$, and thus $\phi^D_{2i+2k-\frac{1}{2}} \phi^D_{-2n+1-\frac{1}{2}}\dots \phi^D_{-3-\frac{1}{2}}\phi^D_{-1-\frac{1}{2}}|0\rangle=0$. If, on the other hand $2i+2k- \frac{1}{2}=-(-2l +1-\frac{1}{2})$ for some $1\leq l\leq n, \ l\in \mathbb{Z}$, then we have
\begin{align*}
r_1(E_{i, i+k}) \phi^D_{-2n+1-\frac{1}{2}}&\dots \phi^D_{-3-\frac{1}{2}}\phi^D_{-1-\frac{1}{2}}|0\rangle \\
&=\phi^D_{-2i+\frac{1}{2}} \phi^D_{2i+2k-\frac{1}{2}}\phi^D_{-2n+1-\frac{1}{2}}\dots \phi^D_{-2l+1-\frac{1}{2}}\dots \phi^D_{ - 1-\frac{1}{2}}|0\rangle\\
&=\pm\phi^D_{-2i+\frac{1}{2}} \phi^D_{-2n+1-\frac{1}{2}}\dots \widehat{\phi^D_{-2l+1-\frac{1}{2}}}\phi^D_{-2l+3-\frac{1}{2}}\dots \phi^D_{ - 1-\frac{1}{2}}|0\rangle;
\end{align*}
here $\widehat{\phi^D_{ -2l+1 -\frac{1}{2}}}$ denotes the fact that $\phi^D_{-2l+1-\frac{1}{2}}$ is absent. But then from $2i+2k- 1/2=-(-2l+1 -\frac{1}{2})$, we have $i+k=l$ and $-2i+\frac{1}{2} =-2l+2k+1-\frac{1}{2}$, and we know from $k \geq 1 $ that either $-2l+3-\frac{1}{2}\leq -2i+\frac{1}{2}\leq -1-\frac{1}{2}$, or $-2i+\frac{1}{2}\geq \frac{1}{2}$. If $-2i+\frac{1}{2}\geq \frac{1}{2}$ (i.e., $k\geq l$), then since $\phi^D_{-2i+\frac{1}{2}}$ anticommutes with all $\phi^D_{-2l+1-\frac{1}{2}}$, we have
\[
\phi^D_{-2i+\frac{1}{2}} \phi^D_{-2n+1-\frac{1}{2}}\dots \widehat{\phi^D_{-2l+1-\frac{1}{2}}}\phi^D_{-2l+3-\frac{1}{2}}\dots \phi^D_{ - 1-\frac{1}{2}}|0\rangle =0.
 \]
 If on the other hand $-2l+3-\frac{1}{2}\leq -2i+\frac{1}{2}\leq  - 1-\frac{1}{2}$, then
that means $\phi^D_{-2i+\frac{1}{2}}=\phi^D_{-2l_1+1-\frac{1}{2}}$ for  $1\leq  l_1\leq l- 1$ and
\begin{align*}
\phi^D_{-2i+\frac{1}{2}}& \phi^D_{-2n+1-\frac{1}{2}}\dots \widehat{\phi^D_{-2l+1-\frac{1}{2}}}\phi^D_{-2l+3-\frac{1}{2}}\dots \phi^D_{ - 1-\frac{1}{2}}|0\rangle\\
&=\phi^D_{-2i+\frac{1}{2}} \phi^D_{-2n+1-\frac{1}{2}}\dots \widehat{\phi^D_{-2l+1-\frac{1}{2}}}\phi^D_{-2l+3-\frac{1}{2}}\dots \phi^D_{-2l_1+1-\frac{1}{2}}\dots \phi^D_{ - 1-\frac{1}{2}}|0\rangle\\
&=\phi^D_{-2l_1+1-\frac{1}{2}} \phi^D_{-2n+1-\frac{1}{2}}\dots \widehat{\phi^D_{-2l+1-\frac{1}{2}}}\phi^D_{-2l+1-\frac{1}{2}}\dots \phi^D_{-2l_1+1-\frac{1}{2}}\dots \phi^D_{ - 1-\frac{1}{2}}|0\rangle\\
&=\pm \phi^D_{-2n+1-\frac{1}{2}}\dots \widehat{\phi^D_{-2l+1-\frac{1}{2}}} \phi^D_{-2l+1-\frac{1}{2}}\dots \phi^D_{-2l_1+1-\frac{1}{2}}\phi^D_{-2l_1+1-\frac{1}{2}}\dots \phi^D_{ - 1-\frac{1}{2}}|0\rangle =0;
\end{align*}
since $\phi^D_{-2l_1+1-\frac{1}{2}}\phi^D_{-2l_1+1-\frac{1}{2}}=0$.

Consider now $v_{-n}=\phi^D_{-2n+2-\frac{1}{2}}\dots \phi^D_{-2-\frac{1}{2}}\phi^D_{-\frac{1}{2}}|0\rangle$, $n>0$.
Again, first we consider the case  $2i+2k>0$, when
$:\phi^D_{-2i+\frac{1}{2}} \phi^D_{2i+2k-\frac{1}{2}}: =\phi^D_{-2i+\frac{1}{2}} \phi^D_{2i+2k-\frac{1}{2}}$.
\begin{equation*}
r_1(E_{i, i+k}) \phi^D_{-2n+2-\frac{1}{2}}\dots \phi^D_{-2-\frac{1}{2}}\phi^D_{-\frac{1}{2}}|0\rangle=\phi^D_{-2i+\frac{1}{2}} \phi^D_{2i+2k-\frac{1}{2}}\phi^D_{-2n+2-\frac{1}{2}}\dots \phi^D_{-2-\frac{1}{2}}\phi^D_{-\frac{1}{2}}|0\rangle.
\end{equation*}
Now since it is impossible to have $2i+2k-\frac{1}{2}=-(-2l +2-\frac{1}{2})$ for any $l\in \mathbb{Z}$, $\phi^D_{2i+2k-\frac{1}{2}}$ will anticommute with any $\phi^D_{-2l-\frac{1}{2}}$, and thus $\phi^D_{2i+2k-\frac{1}{2}}\phi^D_{-2n+2-\frac{1}{2}}\dots \phi^D_{-2-\frac{1}{2}}\phi^D_{-\frac{1}{2}}|0\rangle=0$.

Next, let $2i+2k\leq 0$, then $:\phi^D_{-2i+\frac{1}{2}} \phi^D_{2i+2k-\frac{1}{2}}: = -\phi^D_{2i+2k-\frac{1}{2}}\phi^D_{-2i+\frac{1}{2}}$, and
\begin{equation*}
r_1(E_{i, i+k}) \phi^D_{-2n+2-\frac{1}{2}}\dots \phi^D_{-2-\frac{1}{2}}\phi^D_{-\frac{1}{2}}|0\rangle
=-\phi^D_{2i+2k-\frac{1}{2}}\phi^D_{-2i+\frac{1}{2}} \phi^D_{-2n+2-\frac{1}{2}}\dots \phi^D_{-2-\frac{1}{2}}\phi^D_{-\frac{1}{2}}|0\rangle.
\end{equation*}
Now  $\phi^D_{-2i+\frac{1}{2}}$ is an operator annihilating the vacuum, and unless we have $-2i+ \frac{1}{2}=-(-2l -\frac{1}{2})$ for some $1\leq l\leq n-1, \ l\in \mathbb{Z}$, then it will anticommute with any of the $\phi^D_{-2l-\frac{1}{2}}$, $1\leq l\leq n-1$, and thus $\phi^D_{-2i+\frac{1}{2}}\phi^D_{-2n+2-\frac{1}{2}}\dots \phi^D_{-2-\frac{1}{2}}\phi^D_{-\frac{1}{2}}|0\rangle=0$. If, on the other hand $-2i+ \frac{1}{2}=-(-2l -\frac{1}{2})$ for some $1\leq l\leq n-1, \ l\in \mathbb{Z}$, then we have
\begin{align*}
r_1(E_{i, i+k}) \phi^D_{-2n+2-\frac{1}{2}}&\dots \phi^D_{-2-\frac{1}{2}}\phi^D_{-\frac{1}{2}}|0\rangle \\
&=-\phi^D_{2i+2k-\frac{1}{2}}\phi^D_{-2i+\frac{1}{2}} \phi^D_{-2n+2-\frac{1}{2}}\dots \phi^D_{-2l+2-\frac{1}{2}}\dots \phi^D_{-\frac{1}{2}}|0\rangle\\
&=\pm\phi^D_{2i+2k-\frac{1}{2}} \phi^D_{-2n+2-\frac{1}{2}}\dots \widehat{\phi^D_{-2l-\frac{1}{2}}}\phi^D_{-2l+2-\frac{1}{2}}\dots \phi^D_{-\frac{1}{2}}|0\rangle.
\end{align*}
 But then from $-2i+ 1/2=-(-2l -\frac{1}{2})$, we have $i=-l$ and $2i+2k-\frac{1}{2} =2k-2l-\frac{1}{2}$, and we know from $i+k\leq 0$ that $-2l+2-\frac{1}{2}\leq 2i+2k-\frac{1}{2}\leq -\frac{1}{2}$. That means $\phi^D_{2i+2k-\frac{1}{2}}=\phi^D_{-2l_1-\frac{1}{2}}$ for  $0\leq l_1\leq l-1$   and
\begin{align*}
\phi^D_{2i+2k-\frac{1}{2}}& \phi^D_{-2n+2-\frac{1}{2}}\dots \widehat{\phi^D_{-2l-\frac{1}{2}}}\phi^D_{-2l+2-\frac{1}{2}}\dots \phi^D_{-\frac{1}{2}}|0\rangle\\
&=\phi^D_{2i+2k-\frac{1}{2}} \phi^D_{-2n+2-\frac{1}{2}}\dots \widehat{\phi^D_{-2l-\frac{1}{2}}}\phi^D_{-2l+2-\frac{1}{2}}\dots \phi^D_{-2l_1-\frac{1}{2}}\dots \phi^D_{-\frac{1}{2}}|0\rangle\\
&=\phi^D_{-2l_1-\frac{1}{2}} \phi^D_{-2n+2-\frac{1}{2}}\dots \widehat{\phi^D_{-2l-\frac{1}{2}}}\phi^D_{-2l+2-\frac{1}{2}}\dots \phi^D_{-2l_1-\frac{1}{2}}\dots \phi^D_{-\frac{1}{2}}|0\rangle\\
&=\pm \phi^D_{-2n+2-\frac{1}{2}}\dots \widehat{\phi^D_{-2l-\frac{1}{2}}} \phi^D_{-2l+2-\frac{1}{2}}\dots \phi^D_{-2l_1-\frac{1}{2}}\phi^D_{-2l_1-\frac{1}{2}}\dots \phi^D_{-\frac{1}{2}}|0\rangle =0;
\end{align*}
since $\phi^D_{-2l_1-\frac{1}{2}}\phi^D_{-2l_1-\frac{1}{2}}=0$.
\end{proof}
\begin{lem}\textbf(Calculating the weights)\label{highestweights}
  For all $i,n\in\mathbb Z$,
$r_1(E_{ii})v_n={^a}\hat\Lambda_n(E_{ii})v_n$.
\end{lem}
\begin{proof}  We want to prove
\begin{align*}
&r_1(E_{i, i})v_0=0\cdot v_0;  \\
&r_1(E_{i, i})v_n=1\cdot v_n \quad \text{for}\  0<i\leq n; \quad r_1(E_{i, i})v_n=1\cdot v_n, \quad \text{for} \ i>n\geq 1;\\
&r_1(E_{i, i})v_n=0\cdot v_n\quad \text{for} \ i\leq 0, \ n>0;\\
&r_1(E_{i, i})v_{-n}=-1\cdot v_{-n}\quad \text{for} \ -n+1\leq i\leq 0; \quad r_1(E_{i, i})v_{-n}=0\cdot v_{-n}\quad \text{for} \  -n+1> i > 0;\\
&r_1(E_{i, i})v_{-n}=0\cdot v_{-n}\quad \text{for}\ i>0, \ n>0.
\end{align*}
We have
\begin{align*}
r_1(E_{i, i})=:\phi^D_{-2i+\frac{1}{2}} \phi^D_{2i-\frac{1}{2}}:&=\begin{cases}\phi^D_{-2i+\frac{1}{2}} \phi^D_{2i-\frac{1}{2}},\quad \text{ for $i>0$,}  \\
-  \phi^D_{2i-\frac{1}{2}}\phi^D_{-2i+\frac{1}{2}},\quad \text{ for $i\leq 0$ }
\end{cases} \\
\end{align*}
For $n=0$,  $v_0=|0\rangle$, and $r_1(E_{i, i})= :\phi^D_{-2i+\frac{1}{2}} \phi^D_{2i-\frac{1}{2}}:$, and it is clear that $r_1(E_{i, i})|0\rangle =0$. \\
Let $n>0$ and $v_n=\phi^D_{-2n+1-\frac{1}{2}}\dots \phi^D_{-3-\frac{1}{2}}\phi^D_{-1-\frac{1}{2}}|0\rangle$.
If $i>0$, $\phi^D_{2i-\frac{1}{2}}$ from $r_1(E_{i, i})$ will anticommute with any $\phi^D_{-2l+1-\frac{1}{2}}$ in $v_n$, \textbf{except}  for $l$ such that $-2l+1-\frac{1}{2}=-(2i-\frac{1}{2})$, i.e., $l=i$. In the case there exist an $l$ such that
$l=i$, then
\begin{align*}
r_1(E_{i, i})v_n& =r_1(E_{i, i}) \phi^D_{-2n+1-\frac{1}{2}}\dots \phi^D_{-3-\frac{1}{2}}\phi^D_{-1-\frac{1}{2}}|0\rangle \\
&=\phi^D_{-2i+\frac{1}{2}} \phi^D_{2i-\frac{1}{2}}\phi^D_{-2n+1-\frac{1}{2}}\dots \phi^D_{-2l+1-\frac{1}{2}}\dots \phi^D_{-1-\frac{1}{2}}|0\rangle\\
&=(-1)^{n-l}\phi^D_{-2i+\frac{1}{2}} \phi^D_{-2n+1-\frac{1}{2}}\dots \phi^D_{2i-\frac{1}{2}}\phi^D_{-2l+1-\frac{1}{2}}\dots \phi^D_{-1-\frac{1}{2}}|0\rangle\\
&=(-1)^{n-l}\phi^D_{-2l+1-\frac{1}{2}} \phi^D_{-2n+1-\frac{1}{2}}\dots \widehat{\phi^D_{-2l+1-\frac{1}{2}}}\dots \phi^D_{ - 1-\frac{1}{2}}|0\rangle\\
&= \phi^D_{-2n+1-\frac{1}{2}}\dots \phi^D_{-2l+1-\frac{1}{2}}\dots \phi^D_{ - 1-\frac{1}{2}}|0\rangle =v_n
\end{align*}
Hence $r_1(E_{i, i})v_n=v_n$ for $i\leq 0$ when $n\geq i$ (i.e., exist an $l$ such that
$l=i$), and  $r_1(E_{i, i})v_n=0$ for $i\leq 0$ when $n<i$.

If $i\leq 0$, then $r_1(E_{i, i})=-\phi^D_{2i-\frac{1}{2}} \phi^D_{-2i+\frac{1}{2}}$ and $\phi^D_{-2i+\frac{1}{2}}$ will anticommute with \textbf{any}  $\phi^D_{-2l+1-\frac{1}{2}}$ in $v_{n}$, as it is impossible to have  $-2l+1-\frac{1}{2}=-(-2i+\frac{1}{2})$. Hence
$r_1(E_{i, i})v_n=0$ for $i\leq 0$.

Let $n>0$ and $v_{-n}=\phi^D_{-2n+2-\frac{1}{2}}\dots \phi^D_{-2-\frac{1}{2}}\phi^D_{-\frac{1}{2}}|0\rangle$.
If $i\leq 0$, $\phi^D_{-2i+\frac{1}{2}}$ from $r_1(E_{i, i})$ will anticommute with any $\phi^D_{-2l+2-\frac{1}{2}}$ in $v_{-n}$, \textbf{except}  for $l$ such that $-2l+2-\frac{1}{2}=-(2i+\frac{1}{2})$, i.e., $l=-i-1$. In the case there exist an $l$ such that
$l=-i-1$, then
\begin{align*}
r_1(E_{i, i})v_n& =r_1(E_{i, i}) \phi^D_{-2n+2-\frac{1}{2}}\dots \phi^D_{-2-\frac{1}{2}}\phi^D_{-\frac{1}{2}}|0\rangle \\
&=-\phi^D_{2i-\frac{1}{2}} \phi^D_{-2i+\frac{1}{2}}\phi^D_{-2n+2-\frac{1}{2}}\dots \phi^D_{-2l+2-\frac{1}{2}}\dots \phi^D_{-\frac{1}{2}}|0\rangle\\
&=-(-1)^{n-l}\phi^D_{2i-\frac{1}{2}} \phi^D_{-2n+2-\frac{1}{2}}\dots \phi^D_{-2i+\frac{1}{2}}\phi^D_{-2l+2-\frac{1}{2}}\dots \phi^D_{-\frac{1}{2}}|0\rangle\\
&=-(-1)^{n-l}\phi^D_{-2l+2-\frac{1}{2}} \phi^D_{-2n+2-\frac{1}{2}}\dots \widehat{\phi^D_{-2l+2-\frac{1}{2}}}\dots \phi^D_{-\frac{1}{2}}|0\rangle\\
&= -\phi^D_{-2n+2-\frac{1}{2}}\dots \phi^D_{-2l+2-\frac{1}{2}}\dots \phi^D_{-\frac{1}{2}}|0\rangle =v_n
\end{align*}
Hence $r_1(E_{i, i})v_{-n}=-v_{-n}$ for $i\leq 0$ when $n\geq -i-1$ (i.e., exist an $l$ such that
$l=-i-1$), and  $r_1(E_{i, i})v_{-n}=0$ for $i\leq 0$ when $n<-i-1$.

If $i>0$, then $r_1(E_{i, i})=\phi^D_{-2i+\frac{1}{2}} \phi^D_{2i-\frac{1}{2}}$ and  $\phi^D_{2i-\frac{1}{2}}$ will anticommute with \textbf{any}  $\phi^D_{-2l+2-\frac{1}{2}}$, as it is impossible to have  $-2l+2-\frac{1}{2}=-(2i-\frac{1}{2})$. Hence
$r_1(E_{i, i})v_{-n}=0$ for $i> 0$.
\end{proof}
\begin{prop}\label{generating}  For any $n\in\mathbb Z$, $\mathit{F_{(n)}^{\ten \frac{1}{2}}}$ is an $a_\infty$-submodule of $\mathit{F^{\ten \frac{1}{2}}}$ and $\mathit{F_{(n)}^{\ten \frac{1}{2}}}=U(a_{\infty}^-)v_n$.
\end{prop}
\begin{proof}  First, $\mathit{F_{(n)}^{\ten \frac{1}{2}}}$ is an $a_\infty$-submodule of $\mathit{F^{\ten \frac{1}{2}}}$ as $r_1(E_{ij})$ is a homogenous operator acting on $\mathit{F^{\ten \frac{1}{2}}}$ of degree $0$ with respect to the grading given in \lemref{grading} and in \eqnref{grading2}.
Then certainly that gives us $U(a_{\infty}^-)v_n\subseteq \mathit{F_{(n)}^{\ten \frac{1}{2}}}$.
The proof that  $\mathit{F_{(n)}^{\ten \frac{1}{2}}}=U(a_{\infty}^-)v_n$ is similar for each $n\in \mathbb{Z}$, thus we will show it only for $n=0$. Let $v\in \mathit{F_{(0)}^{\ten \frac{1}{2}}}$, without loss of generality here we can assume $v$ is homogeneous, i.e. $v=\phi^D_{-n_k-\frac{1}{2}}\dots \phi^D_{-n_2-\frac{1}{2}}\phi^D_{-n_1-\frac{1}{2}}|0\rangle$. Since $v\in \mathit{F_{(0)}^{\ten \frac{1}{2}}}$, we have $k=2l$ and precisely half of the indexes $n_1, n_2, \dots, n_{k}$ are even, the other half are odd.
Thus we can write after eventual use of the anticommutation relations in $\mathit{Cl_D}$
\[
v=\pm \phi^D_{-n_l^o-\frac{1}{2}}\phi^D_{-n_l ^{e} -\frac{1}{2}}\cdots \phi^D_{-n_1^o-\frac{1}{2}}\phi^D_{-n_1^e-\frac{1}{2}}|0\rangle,
\]
i.e., we have rearranged the factors in pairs, so that the indexes in the pairs are $n_s^e$ and $n_s^o$, and we have $n_l^o>\dots >n_1^o$, $n_l^e>\dots >n_1^e$.
Hence we have $n_s^e=2 q_s$ for some $q_s\in \mathbb{Z}$ and $n_s^o=2p_s-1$ for some $p_s\in \mathbb{Z}$.
Thus
\begin{equation*}
v=\pm \phi^D_{-n_l^o-\frac{1}{2}}\phi^D_{-n_l^e-\frac{1}{2}}\cdots \phi^D_{-n_1^o-\frac{1}{2}}\phi^D_{-n_1^e-\frac{1}{2}}|0\rangle\\
=\pm r_1(E_{p_l, -q_l})\cdots r_1(E_{p_1, -q_1}|0\rangle;
\end{equation*}
which proves that $\mathit{F_{(0)}^{\ten \frac{1}{2}}}=U(a_{\infty}^-)v_0$. The proof that $\mathit{F_{(n)}^{\ten \frac{1}{2}}}=U(a_{\infty}^-)v_n$ is very similar for the other $n\in \mathbb{Z}$ and we will omit it.
\end{proof}
\begin{prop}\label{irreducibleFn}
 For any $n\in \mathbb{Z}$, \ $\mathit{F_{(n)}^{\ten \frac{1}{2}}}$  is an irreducible submodule for the representation $r_1$ of $a_{\infty}$ inside $\mathit{F^{\ten \frac{1}{2}}}$.
Moreover for any $n\in\mathbb Z$, $\mathit{F_{(n)}^{\ten \frac{1}{2}}}\cong L( a_\infty; ^a\!\hat\Lambda_n, 1)$.
\end{prop}
\begin{proof}
This proof uses the uniqueness property of the contragradient Hermitian symmetric form on the Verma module $V( a_\infty; ^a\!\hat\Lambda_n, 1)$ (see \cite{Jantzen}, \cite{Kac-Lie} or \cite{Moody}). For a more direct, but calculational proof the reader can see the Appendix.
   It is well known that one can define $\omega$ the conjugate linear involutive anti-automorphism on $\mathit{Cl_D}$ by $\omega(\phi_m^D)=\phi_{-m}^D$ for all $m\in\mathbb Z +\frac{1}{2}$ and $\omega(1)=  1$.   Recall the module $\mathit{F^{\ten \frac{1}{2}}}$ is defined to be the induced module
$$
\mathit{F^{\ten \frac{1}{2}}}=\mathit{Cl_D}\otimes_{\mathbb C  1\otimes \mathit{Cl_D}_+}\mathbb C|0\rangle
$$
whereby $ 1|0\rangle=|0\rangle$ and $ \mathit{Cl_D}_+|0\rangle=0$.   The conjugate linear involutive anti-automorphism (or antilinear antiautomorphism) $\omega:\mathit{Cl_D}\to \mathit{Cl_D}$ defined by $\omega(\phi_{n})=\phi_{-n}$  gives rise to a non-degenerate positive definite form $\langle \enspace|\enspace\rangle$ defined on $\mathit{F^{\ten \frac{1}{2}}}$ whereby
$$
\langle Xv\,|\,w\rangle = \langle v\,|\,\omega(X)w\rangle
$$
for all $v,w\in\mathit{F^{\ten \frac{1}{2}}}$ and $X\in \mathit{Cl_D}$.    Observe that $\mathit{F^{\ten \frac{1}{2}}}_{(n)}\perp \mathit{F^{\ten \frac{1}{2}}}_{(m)}$ for $m\neq n$.


It is straightforward to check that
\begin{align}
\omega(r_1(E_{m,n}))
&=r_1(E_{n,m})=r_1(\omega(E_{m,n})). \label{intertwinningid}
\end{align}
Thus $\omega$ defined on $\mathit{Cl_D}$ agrees with the compact anti-involution defined on $a_\infty$ given earlier.

We have from  \propref{singularvector}, \lemref{highestweights}, and \propref{generating} that $\mathit{F_{(n)}^{ \ten\frac{1}{2}}}$ is a highest weight $a_\infty$-module.  By the universal mapping property of Verma modules there exists an $a_\infty$-module homomorphism $\pi:V( a_\infty; ^a\!\hat\Lambda_n, 1)\to \mathit{F_{(n)}^{ \ten\frac{1}{2}}}$ sending the highest weight vector $v_{^a\!\hat\Lambda_n}$ of the Verma module $V( a_\infty; ^a\!\hat\Lambda_n, 1)$ to  $v_n$ in $\mathit{F_{(n)}^{ \ten\frac{1}{2}}}$.
Moreover the Hermitian symmetric $\omega$-contragradient form $\langle\enspace,\enspace\rangle$ on $\mathit{F_{(n)}^{ \ten\frac{1}{2}}}$ pulls back to a Hermitian symmetric form $(\enspace,\enspace)$ on  the Verma module $V( a_\infty; ^a\!\hat\Lambda_n, 1)$.   In other words
\begin{equation}
(uv_{^a\!\hat\Lambda_n},u'v_{^a\!\hat\Lambda_n}):=\langle \pi(uv_{^a\!\hat\Lambda_n}),\pi(u'v_{^a\!\hat\Lambda_n})\rangle=\langle r_1(u)v_n,r_1(u')v_n\rangle.
\end{equation}
for all $u,u'\in a_\infty$.
By \eqnref{intertwinningid}, we have
\begin{align*}
(Xuv_{^a\!\hat\Lambda_n},u'v_{^a\!\hat\Lambda_n})&=\langle r_1(X)r_1(u)v_n,r_1(u')v_n\rangle\\
&=\langle r_1(u)v_n,\omega(r_1(X))r_1(u')v_n\rangle\\
&=\langle r_1(u)v_n,r_1(\omega(X) u')v_n\rangle\\
&=(uv_{^a\!\hat\Lambda_n},\omega(X)u'v_{^a\!\hat\Lambda_n})
\end{align*}
Hence $(\enspace,\enspace)$ is contragradient with respect to $\omega$.

Now it is known that there is a unique contragradient Hermitian symmetric form on the Verma module $V( a_\infty; ^a\!\hat\Lambda_n, 1)$ with $(v_{^a\!\hat\Lambda_n},v_{^a\!\hat\Lambda_n})=\langle v_n, v_n\rangle$  and its radical is precisely the unique maximal submodule $\overline{V( a_\infty; ^a\!\hat\Lambda_n, 1)}$ of $V( a_\infty; ^a\!\hat\Lambda_n, 1)$ (see \cite{Jantzen}, \cite{Kac-Lie} or \cite{Moody}).      So
\begin{equation}
0=(\overline{V( a_\infty; ^a\!\hat\Lambda_n, 1)},u'v_{^a\!\hat\Lambda_n}):=\langle \pi(\overline{V( a_\infty; ^a\!\hat\Lambda_n, 1)} ),\pi(u'v_{^a\!\hat\Lambda_n})\rangle\end{equation}
for all $u'\in U(a_\infty)$.   Since $\langle \enspace,\enspace \rangle$ is nondegenerate and $\pi$ is surjective one must have $\overline{V( a_\infty; ^a\!\hat\Lambda_n, 1)}\subseteq\ker \pi$.
 Thus $L( a_\infty; ^a\!\hat\Lambda_n, 1)=V( a_\infty; ^a\!\hat\Lambda_n, 1)/\overline{V( a_\infty; ^a\!\hat\Lambda_n, 1)}\cong \mathit{F^{\ten \frac{1}{2}}}_{(n)}$.
 \end{proof}
\begin{remark}
An alternative proof of the irreducibility of  $\mathit{F_{(n)}^{\ten \frac{1}{2}}}$ can be given as follows. In \cite{AngTVA} we considered the Heisenberg algebra  $\mathcal{H}_{\mathbb{Z}}$, which is a subalgebra of   $a_{\infty}$ represented by
\[
h_n\mapsto \sum_{i\in \mathbb{Z}} E_{i, i+n}
\]
In \cite{AngTVA} we prove that each $ \mathit{F_{(n)}^{\ten \frac{1}{2}}}$ is irreducible under $\mathcal{H}_{\mathbb{Z}}$. Since $\mathcal{H}_{\mathbb{Z}}$ is a subalgebra of   $a_{\infty}$, then $ \mathit{F_{(n)}^{\ten \frac{1}{2}}}$ is irreducible under $a_{\infty}$.
\end{remark}
The previous three propositions can now be combined in the following
\begin{thm}\label{thm:main}
As $a_{\infty}$ modules with central charge $c=1$ $ \mathit{F_{(n)}^{\ten \frac{1}{2}}}\cong L( a_\infty; ^a\!\hat\Lambda_n, 1)$ and  $\mathit{F^{\ten \frac{1}{2}}}$ decomposes into irreducible submodules as follows:
\[
\mathit{F^{\ten \frac{1}{2}}}\cong \oplus_{n\in \mathbb{Z}} L(a_\infty; ^a\!\hat\Lambda_n, 1).
\]
Hence $\mathit{F^{\ten \frac{1}{2}}}\cong \mathit{F^{\ten 1}}$ as $a_{\infty}$ modules with central charge $c=1$.
\end{thm}
Here $\mathit{F^{\ten 1}}$ denotes the fermionic Fock space of the charged free fermions (see Remark \ref{remark:F^1}  above, after \cite{Frenkel-BF},  \cite{KacRaina}, \cite{Wang}, \cite{Kac}, \cite{WangKac}, \cite{WangDual}).

Since $\mathit{F^{\ten \frac{1}{2}}}\cong \mathit{F^{\ten 1}}$ as $a_{\infty}$ modules with central charge $c=1$, we can use Theorem 3.2 of \cite{WangDuality} to get the decomposition of the irreducible central charge $c=\frac{1}{2}$ $d_{\infty}$ modules $\mathit{F_{\bar{0}}^{\ten \frac{1}{2}}}$ and $\mathit{F_{\bar{1}}^{\ten \frac{1}{2}}}$  in terms of the new $d_{\infty}$ central charge $c=1$  action:
\begin{cor}
As $d_{\infty}$ modules with central charge $c=1$  we have
\begin{align}
&\mathit{F_{(n)}^{\ten \frac{1}{2}}}\cong L( d_\infty; ^{\frak{dd}}\!\hat\Lambda_n, 1)\quad \text{for} \ n\neq 0
&\mathit{F_{(0)}^{\ten \frac{1}{2}}}\cong L( d_\infty; ^{\frak{dd}}\!\hat\Lambda_0, 1)\oplus L( d_\infty; ^{\frak{dd}}\!\hat\Lambda_{det}, 1);
\end{align}
where for $n\neq 0$  $\mathit{F_{(n)}^{\ten \frac{1}{2}}}$ has highest weight vector $v_n$, and $\mathit{F_{(0)}^{\ten \frac{1}{2}}}$ decomposes into two irreducible highest weight modules, with highest weight vectors $v_0$ and $\tilde{v}_0=\phi^D_{-\frac{3}{2}}\phi^D_{-\frac{1}{2}}|0\rangle\in \mathit{F_{(0)}^{\ten \frac{1}{2}}}$ (note $\tilde{v}_0$ is not a highest weight vector for the $a_{\infty}$ action, only the $d_{\infty}$ action).
Thus as $d_{\infty}$ modules with central charge $c=1$
\[
\mathit{F^{\ten \frac{1}{2}}}=  \left(\oplus_{\substack{n\in \mathbb{Z}\\ n\neq 0}}L( d_\infty; ^{\frak{dd}}\!\hat\Lambda_n, 1)\right)\bigoplus \left(L( d_\infty; ^{\frak{dd}}\!\hat\Lambda_0, 1)\oplus L( d_\infty; ^{\frak{dd}}\!\hat\Lambda_{det}, 1)\right),
\]
and
\begin{align*}
\mathit{F_{\bar{1}}^{\ten \frac{1}{2}}}&= \oplus_{n\in \mathbb{Z}}L( d_\infty; ^{\frak{dd}}\!\hat\Lambda_{2n-1}, 1);\\
\mathit{F_{\bar{0}}^{\ten \frac{1}{2}}}&= \left(\oplus_{\substack{n\in \mathbb{Z}\\ n\neq 0}}L( d_\infty; ^{\frak{dd}}\!\hat\Lambda_{2n}, 1)\right)\bigoplus \left(L( d_\infty; ^{\frak{dd}}\!\hat\Lambda_0, 1)\oplus L( d_\infty; ^{\frak{dd}}\!\hat\Lambda_{det}, 1)\right).
\end{align*}
\end{cor}
Here, as in \cite{WangDuality}, the highest weights $^{\frak{dd}}\!\hat\Lambda_n$ are obtained from the restrictions of $^a\!\hat\Lambda_n$, except for $^{\frak{dd}}\!\hat\Lambda_{det}$ which is defined by
\begin{equation*}
 ^{\frak{dd}}\!\hat\Lambda_{det}(E_{i,i}- E_{1-i,1-i})=2; \quad \text{for} \ i=1,\quad  ^{\frak{dd}}\!\hat\Lambda_{det}(E_{i,i}- E_{1-i,1-i})=0\quad \text{for}\ i\neq 0, 1.
\end{equation*}

\begin{proof}
Follows directly from Theorem 3.2 of \cite{WangDuality} and Theorem \ref{thm:main}.
\end{proof}

\section{Appendix}

For readers who would like to see a more computational proof of \propref{irreducibleFn} we present one below.

{\it Alternate proof of \propref{irreducibleFn}}:
First we will prove that  the action of $a_{\infty}$ on $\mathit{F^{\ten \frac{1}{2}}}$ will preserve the $\mathbb{Z}$ grading $gd$; that will show that each $\mathit{F_{(n)}^{\ten \frac{1}{2}}}$ is a submodule for the $r_1$ action of $a_{\infty}$ on $\mathit{F^{\ten \frac{1}{2}}}$. Then we  will show that the submodule $\mathit{F_{(n)}^{\ten \frac{1}{2}}}$ is irreducible.

Let $v=\phi^D_{-n_k-\frac{1}{2}}\dots \phi^D_{-n_2-\frac{1}{2}}\phi^D_{-n_1-\frac{1}{2}}|0\rangle$ be any homogeneous vector in  $\mathit{F_{(n)}^{\ten \frac{1}{2}}}$. We have
\begin{align*}
r_1(E_{p, q})=:\phi^D_{-2p+\frac{1}{2}} \phi^D_{2q-\frac{1}{2}}:&=\begin{cases}\phi^D_{-2p+\frac{1}{2}} \phi^D_{2q-\frac{1}{2}},\quad \text{ unless} \  q\leq 0\ \text{and}\ 2q-\frac{1}{2}=-(-2p+ \frac{1}{2})  \\
-  \phi^D_{2q-\frac{1}{2}}\phi^D_{-2p+\frac{1}{2}},\quad \text{otherwise.  }
\end{cases}
\end{align*}
Consider first the case when $-(2q-\frac{1}{2})=-2p+ \frac{1}{2}$ and $q\leq 0$. Then $-2p+\frac{1}{2}>0$, and  in
\begin{equation*}
r_1(E_{p, q})v= -  \phi^D_{2q-\frac{1}{2}}\phi^D_{-2p+\frac{1}{2}}\phi^D_{-n_k-\frac{1}{2}}\dots \phi^D_{-n_2-\frac{1}{2}}\phi^D_{-n_1-\frac{1}{2}}|0\rangle
\end{equation*}
$\phi^D_{-2p+\frac{1}{2}}$ will either anticommute with all $\phi^D_{-n_s-\frac{1}{2}}$, $s=1, \dots , k$, in which case $r_1(E_{p, q})v=0$ as $\phi^D_{-2p+\frac{1}{2}}|0\rangle =0$; or otherwise we will have
$-2p+\frac{1}{2}=-(-n_s-\frac{1}{2})$, for some $s=1, \dots , k$. In that case we also have
$-n_s-\frac{1}{2}=2q-\frac{1}{2}$ and
\begin{align*}
&r_1(E_{p, q})v= -  \phi^D_{2q-\frac{1}{2}}\phi^D_{-2p+\frac{1}{2}}\phi^D_{-n_k-\frac{1}{2}}\dots \phi^D_{-n_s-\frac{1}{2}}\dots \phi^D_{-n_1-\frac{1}{2}}|0\rangle\\
&=\pm  \phi^D_{2q-\frac{1}{2}}\phi^D_{-n_k-\frac{1}{2}}\dots \widehat{\phi^D_{-n_s-\frac{1}{2}}}\dots \phi^D_{-n_1-\frac{1}{2}}|0\rangle\\
&=\pm \phi^D_{-n_k-\frac{1}{2}}\dots \phi^D_{-n_s-\frac{1}{2}}\dots \phi^D_{-n_1-\frac{1}{2}}|0\rangle.
\end{align*}
This shows that in both cases when $-(2q-\frac{1}{2})=-2p+ \frac{1}{2}$ we have if  $v\in \mathit{F_{(n)}^{\ten \frac{1}{2}}}$ then  $r_1(E_{p, q})v\in \mathit{F_{(n)}^{\ten \frac{1}{2}}}$.

Next, consider the case when $-(2q-\frac{1}{2})\neq -2p+ \frac{1}{2}$, but still $q\leq 0$, which implies again that $r_1(E_{p, q})=-  \phi^D_{2q-\frac{1}{2}}\phi^D_{-2p+\frac{1}{2}}$ (anticommutation) and again we consider the possible cases for
\begin{equation*}
r_1(E_{p, q})v= -  \phi^D_{2q-\frac{1}{2}}\phi^D_{-2p+\frac{1}{2}}\phi^D_{-n_k-\frac{1}{2}}\dots \phi^D_{-n_2-\frac{1}{2}}\phi^D_{-n_1-\frac{1}{2}}|0\rangle.
\end{equation*}
The first case is when $p\leq 0$ and $\phi^D_{-2p+\frac{1}{2}}$  anticommutes with all $\phi^D_{-n_s-\frac{1}{2}}$, $s=1, \dots , k$, in which case $r_1(E_{p, q})v=0$. The second case is again as above when $-2p+\frac{1}{2}=-(-n_s-\frac{1}{2})$, for some $s=1, \dots , k$. We again have
\begin{align*}
&r_1(E_{p, q})v= -  \phi^D_{2q-\frac{1}{2}}\phi^D_{-2p+\frac{1}{2}}\phi^D_{-n_k-\frac{1}{2}}\dots \phi^D_{-n_s-\frac{1}{2}}\dots \phi^D_{-n_1-\frac{1}{2}}|0\rangle\\
&=\pm  \phi^D_{2q-\frac{1}{2}}\phi^D_{-n_k-\frac{1}{2}}\dots \widehat{\phi^D_{-n_s-\frac{1}{2}}}\dots \phi^D_{-n_1-\frac{1}{2}}|0\rangle.
\end{align*}
In other words, we have ``removed" $\phi^D_{-n_s-\frac{1}{2}}$ with \textbf{even} index $n_s=2p$.
Now since $q<0$, we either have $2q-\frac{1}{2}=-n_{t}-\frac{1}{2}$ for some $t=1, \dots , k, t\neq s$ in which case again we have $r_1(E_{p, q})v=0$ as $\phi^D_{-2q-\frac{1}{2}} \phi^D_{-2q-\frac{1}{2}}=0$. Or otherwise we have ``added" $\phi^D_{-n_t-\frac{1}{2}}$ with \textbf{even} index $n_t=-2q$, which doesn't change the degree $dg$, as we have first ``removed" an even index  $n_s=2p$ and then ``added" an even index $n_t=-2q$. Finally, if we have both $q<0$ and $p\geq 1$, then
\begin{equation*}
r_1(E_{p, q})v= -  \phi^D_{2q-\frac{1}{2}}\phi^D_{-2p+1-\frac{1}{2}}\phi^D_{-n_k-\frac{1}{2}}\dots \phi^D_{-n_s-\frac{1}{2}}\dots \phi^D_{-n_1-\frac{1}{2}}|0\rangle;
\end{equation*}
i.e., we observe that the action $r_1(E_{p, q})v$  "adds" both an  \textbf{even} and an \textbf{odd} index to $v$, or annihilates it.
Thus in all cases when
$q<0$ we have if  $v\in \mathit{F_{(n)}^{\ten \frac{1}{2}}}$ then  $r_1(E_{p, q})v\in \mathit{F_{(n)}^{\ten \frac{1}{2}}}$.

Lastly, let $q\geq 1$.
Similar argument as above show that either $r_1(E_{p, q})v=0$, or
\begin{align*}
&r_1(E_{p, q})v=  \phi^D_{-2p+\frac{1}{2}} \phi^D_{2q-\frac{1}{2}} \phi^D_{-n_k-\frac{1}{2}}\dots \phi^D_{-n_t-\frac{1}{2}}\dots \phi^D_{-n_1-\frac{1}{2}}|0\rangle\\
&=\pm  \phi^D_{-2p+1-\frac{1}{2}}\phi^D_{-n_k-\frac{1}{2}}\dots \widehat{\phi^D_{-n_t-\frac{1}{2}}}\dots \phi^D_{-n_1-\frac{1}{2}}|0\rangle;
\end{align*}
where $2q-\frac{1}{2}=-(-n_{t}-\frac{1}{2})$ for some $t=1, \dots , k, t\neq s$. Thus  $n_t=2q-1$ and we have ``removed" an \textbf{odd} index $n_t$. Now if $-2p+1-\frac{1}{2}<0$, we are ``adding" back  an \textbf{odd} index $-2p+1$. If, on the other hand $-2p+1-\frac{1}{2}>0$, then either we get 0, or we remove also an  \textbf{even} index. Thus as a summary, in all cases we either remove an odd (even) index and add an odd (even) index back; we remove both an odd and an even index; or we get 0. Hence in all cases the action of $r_1(E_{p, q})$ will preserve the grading. Hence each $\mathit{F_{(n)}^{\ten \frac{1}{2}}}$ is a submodule for the representation $r_1$ of $a_{\infty}$ on $\mathit{F^{\ten \frac{1}{2}}}$.

Further, from the observations above we can summarize the action $r_1$ on $\mathit{F^{\ten \frac{1}{2}}}$ as follows. We have 3 nontrivial cases: in case 1 $r_1(E_{p, q})$ acting on a homogeneous vector $v$ ``\textbf{adds}" two factors   $\phi^D_{-n_{s_1}-\frac{1}{2}}$ and $\phi^D_{-n_{s_2}-\frac{1}{2}}$, so that one of the indexes $n_{s_1}$, $n_{s_2}$ is \textbf{even} , the other is \textbf{odd} (we will call it for short ``adding an even and an odd index"). In case 2, we \textbf{replace} a factor $\phi^D_{-n_{s_1}-\frac{1}{2}}$ with another factor $\phi^D_{-n_{s_2}-\frac{1}{2}}$, where either both factors $n_{s_1}$, $n_{s_2}$ are \textbf{even}, or both factors $n_{s_1}$, $n_{s_2}$ are \textbf{odd} (we will call it for short ``replacing even with even index" and ``replacing odd with odd index"). And case 3 is when we ``\textbf{remove}" two factors   $\phi^D_{-n_{s_1}-\frac{1}{2}}$ and $\phi^D_{-n_{s_2}-\frac{1}{2}}$, so that one of the indexes $n_{s_1}$, $n_{s_2}$ is \textbf{even} , the other is \textbf{odd} (``removing an even and an odd index"). Note that ``adding an index" that is already present, or ``removing an index" that was absent, will of course produce the 0 vector.

Now we want to prove that for each $n\in \mathbb{Z}$ $\mathit{F_{(n)}^{\ten \frac{1}{2}}}$ is an irreducible module for $a_{\infty}$. This will be done in  two steps. The first step is to prove that each vector in $\mathit{F_{(n)}^{\ten \frac{1}{2}}}$ can be generated from the "$n$-th vacuum vector" $v_n$, i.e., $\mathit{F_{(n)}^{\ten \frac{1}{2}}}=U(a_{\infty}^-)v_n$. This was  done in Proposition \ref{generating}. The second step is to prove that for any vector $v\in \mathit{F_{(n)}^{\ten \frac{1}{2}}}$, we have $v_n\in U(a_{\infty})v$.
Let then  $v\in \mathit{F_{(n)}^{\ten \frac{1}{2}}}$ be any vector, not necessary homogeneous:
 $v=\sum _{hk} c_{hk}v^{hk}$, where $v^{hk}$ are homogeneous vectors, $v^{hk}= \phi^D_{-n_k-\frac{1}{2}}\dots \phi^D_{-n_2-\frac{1}{2}}\phi^D_{-n_1-\frac{1}{2}}|0\rangle$, $c_k\in \mathbb{C}$.
 It is clear that by the operation ``adding an even and an odd index" we can reduce $v$ only to a linear combination of homogeneous vectors with minimal possible length $\tilde{L}$ among the $v^{hk}$ by the following two steps: we would ``add an even and an odd index" starting from the already existing indexes in the vector $v^{hk}$ with the largest length $\tilde{L}$ (which will annihilate it), and then we would remove the same combination back (which will bring the remaining vectors with lower length back to their original length). This is always possible, as any two  lengths within $\mathit{F_{(n)}^{\ten \frac{1}{2}}}$ differ always by an even number; and we always have at least two differing elements $\phi^D_{-n_k-\frac{1}{2}}$  between the homogeneous vectors with the two consecutive lengths, as well as two different elements $\phi^D_{-n_k-\frac{1}{2}}$  between homogeneous vectors with same lengths. Now among the remaining different vectors $v^{hk}$ with minimal possible length $\tilde{L}$ we can use the operations of ``replacing even index with even" and ``replacing odd index with odd" until only a single  homogeneous vector  $v^{hk}$ with largest index $n_k$ the  minimal possible remains. We will show how this algorithm on an example of  a vector $v\in \mathit{F_{(0)}^{\ten \frac{1}{2}}}$ which will illustrate these 3 steps,
 \begin{align*}
 v&=\phi^D_{-7-\frac{1}{2}}\phi^D_{-4-\frac{1}{2}}\phi^D_{-3-\frac{1}{2}}\phi^D_{-2-\frac{1}{2}}|0\rangle + \phi^D_{-5-\frac{1}{2}}\phi^D_{-3-\frac{1}{2}}\phi^D_{-2-\frac{1}{2}}\phi^D_{-0-\frac{1}{2}}|0\rangle \\ &\quad  +\phi^D_{-8-\frac{1}{2}}\phi^D_{-5-\frac{1}{2}}|0\rangle + \phi^D_{-6-\frac{1}{2}}\phi^D_{-1-\frac{1}{2}}|0\rangle +\phi^D_{-5-\frac{1}{2}}\phi^D_{-4-\frac{1}{2}}|0\rangle +\phi^D_{-9-\frac{1}{2}}\phi^D_{-8-\frac{1}{2}}|0\rangle.
 \end{align*}
 We have Step 1:
\begin{align*}
 E_{4, -2}v&=\phi^D_{-7-\frac{1}{2}}\phi^D_{-4-\frac{1}{2}}v= \phi^D_{-7-\frac{1}{2}}\phi^D_{-4-\frac{1}{2}}\phi^D_{-7-\frac{1}{2}}\phi^D_{-4-\frac{1}{2}}\phi^D_{-3-\frac{1}{2}}\phi^D_{-2-\frac{1}{2}}|0\rangle \\
 &\quad + \phi^D_{-7-\frac{1}{2}}\phi^D_{-4-\frac{1}{2}}\phi^D_{-5-\frac{1}{2}}\phi^D_{-3-\frac{1}{2}}\phi^D_{-2-\frac{1}{2}}\phi^D_{-0-\frac{1}{2}}|0\rangle \\ &\quad  +\phi^D_{-7-\frac{1}{2}}\phi^D_{-4-\frac{1}{2}}\phi^D_{-8-\frac{1}{2}}\phi^D_{-5-\frac{1}{2}}|0\rangle + \phi^D_{-7-\frac{1}{2}}\phi^D_{-4-\frac{1}{2}}\phi^D_{-6-\frac{1}{2}}\phi^D_{-1-\frac{1}{2}}|0\rangle \\ &\quad +\phi^D_{-7-\frac{1}{2}}\phi^D_{-4-\frac{1}{2}}\phi^D_{-5-\frac{1}{2}}\phi^D_{-4-\frac{1}{2}}|0\rangle +\phi^D_{-7-\frac{1}{2}}\phi^D_{-4-\frac{1}{2}}\phi^D_{-9-\frac{1}{2}}\phi^D_{-8-\frac{1}{2}}|0\rangle\\
&=\phi^D_{-7-\frac{1}{2}}\phi^D_{-4-\frac{1}{2}}\phi^D_{-5-\frac{1}{2}}\phi^D_{-3-\frac{1}{2}}\phi^D_{-2-\frac{1}{2}}\phi^D_{-0-\frac{1}{2}}|0\rangle
\\
&\quad + \phi^D_{-7-\frac{1}{2}}\phi^D_{-4-\frac{1}{2}}\phi^D_{-8-\frac{1}{2}}\phi^D_{-5-\frac{1}{2}}|0\rangle + \phi^D_{-7-\frac{1}{2}}\phi^D_{-4-\frac{1}{2}}\phi^D_{-6-\frac{1}{2}}\phi^D_{-1-\frac{1}{2}}|0\rangle\\
&\quad +\phi^D_{-7-\frac{1}{2}}\phi^D_{-4-\frac{1}{2}}\phi^D_{-9-\frac{1}{2}}\phi^D_{-8-\frac{1}{2}}|0\rangle;
 \end{align*}
Step 2:
\begin{align*}
 E_{-2, 4} E_{4, -2}v& =\phi^D_{4+\frac{1}{2}}\phi^D_{7+\frac{1}{2}}\phi^D_{-7-\frac{1}{2}}\phi^D_{-4-\frac{1}{2}}v\\
 &=\phi^D_{4+\frac{1}{2}}\phi^D_{7+\frac{1}{2}}\phi^D_{-7-\frac{1}{2}}\phi^D_{-4-\frac{1}{2}}\phi^D_{-5-\frac{1}{2}}\phi^D_{-3-\frac{1}{2}}\phi^D_{-2-\frac{1}{2}}\phi^D_{-0-\frac{1}{2}}|0\rangle
\\
&\quad + \phi^D_{4+\frac{1}{2}}\phi^D_{7+\frac{1}{2}}\phi^D_{-7-\frac{1}{2}}\phi^D_{-4-\frac{1}{2}}\phi^D_{-8-\frac{1}{2}}\phi^D_{-5-\frac{1}{2}}|0\rangle \\
&\quad +\phi^D_{4+\frac{1}{2}}\phi^D_{7+\frac{1}{2}}\phi^D_{-7-\frac{1}{2}}\phi^D_{-4-\frac{1}{2}}\phi^D_{-6-\frac{1}{2}}\phi^D_{-1-\frac{1}{2}}|0\rangle\\
&\quad +\phi^D_{4+\frac{1}{2}}\phi^D_{7+\frac{1}{2}}\phi^D_{-7-\frac{1}{2}}\phi^D_{-4-\frac{1}{2}}\phi^D_{-9-\frac{1}{2}}\phi^D_{-8-\frac{1}{2}}|0\rangle\\
&=\phi^D_{-5-\frac{1}{2}}\phi^D_{-3-\frac{1}{2}}\phi^D_{-2-\frac{1}{2}}\phi^D_{-0-\frac{1}{2}}|0\rangle \\ &\quad +\phi^D_{-8-\frac{1}{2}}\phi^D_{-5-\frac{1}{2}}|0\rangle
 + \phi^D_{-6-\frac{1}{2}}\phi^D_{-1-\frac{1}{2}}|0\rangle+\phi^D_{-9-\frac{1}{2}}\phi^D_{-8-\frac{1}{2}}|0\rangle;
 \end{align*}
We repeat Step 1:
\begin{align*}
 E_{3, -1}E_{-2, 4} E_{4, -2}v& =\phi^D_{-5-\frac{1}{2}}\phi^D_{-2-\frac{1}{2}}E_{-2, 4} E_{4, -2}v\\
 &=\phi^D_{-5-\frac{1}{2}}\phi^D_{-2-\frac{1}{2}}\phi^D_{-5-\frac{1}{2}}\phi^D_{-3-\frac{1}{2}}\phi^D_{-2-\frac{1}{2}}\phi^D_{-0-\frac{1}{2}}|0\rangle \\
 &\quad +\phi^D_{-5-\frac{1}{2}}\phi^D_{-2-\frac{1}{2}}\phi^D_{-8-\frac{1}{2}}\phi^D_{-5-\frac{1}{2}}|0\rangle +\phi^D_{-5-\frac{1}{2}}\phi^D_{-2-\frac{1}{2}}\phi^D_{-6-\frac{1}{2}}\phi^D_{-1-\frac{1}{2}}|0\rangle\\
 &\quad +\phi^D_{-5-\frac{1}{2}}\phi^D_{-2-\frac{1}{2}}\phi^D_{-9-\frac{1}{2}}\phi^D_{-8-\frac{1}{2}}|0\rangle\\
 &=\phi^D_{-5-\frac{1}{2}}\phi^D_{-2-\frac{1}{2}}\phi^D_{-6-\frac{1}{2}}\phi^D_{-1-\frac{1}{2}}|0\rangle +\phi^D_{-5-\frac{1}{2}}\phi^D_{-2-\frac{1}{2}}\phi^D_{-9-\frac{1}{2}}\phi^D_{-8-\frac{1}{2}}|0\rangle.
 \end{align*}
 We repeat Step 2:
 \begin{equation*}
 E_{-1, 3}E_{3, -1}E_{-2, 4} E_{4, -2}v=\phi^D_{-6-\frac{1}{2}}\phi^D_{-1-\frac{1}{2}}|0\rangle +\phi^D_{-9-\frac{1}{2}}\phi^D_{-8-\frac{1}{2}}|0\rangle.
 \end{equation*}
 Finally Step 3:
  \begin{align*}
 E_{1, 5}E_{-1, 3}E_{3, -1}E_{-2, 4} E_{4, -2}v&=\phi^D_{-1-\frac{1}{2}}\phi^D_{9+\frac{1}{2}}\phi^D_{-6-\frac{1}{2}}\phi^D_{-1-\frac{1}{2}}|0\rangle +\phi^D_{-1-\frac{1}{2}}\phi^D_{9+\frac{1}{2}}\phi^D_{-9-\frac{1}{2}}\phi^D_{-8-\frac{1}{2}}|0\rangle\\
 &=\phi^D_{-8-\frac{1}{2}}\phi^D_{-1-\frac{1}{2}}|0\rangle.
 \end{align*}
 Hence similarly we can reduce any  potentially nonhomogeneous vector $v$ we started with to a homogeneous vector $v^{hk}$ by successive  action of  $r_1(E_{pq})$ with various $p, q\in \mathbb{Z}$. Thus we now need only consider a homogeneous vector $v^{hk}=\phi^D_{-n_k-\frac{1}{2}}\dots \phi^D_{-n_2-\frac{1}{2}}\phi^D_{-n_1-\frac{1}{2}}|0\rangle$.
 If $v^{hk}$ has length $\tilde{L}(v^{hk})>|n|$, then $\tilde{L}(v^{hk})=|n|+2l$ and we can use the the operation of ``removing an even and an odd index" $l$ times in succession until we get a vector of minimal length $\tilde{L}=|n|$.
 After that we just have to eventually ``replace some even indexes with even" and ``replace some odd indexes with odd" to produce the highest weight vector $v_n$.
 Hence, we have proved that $v_n\in U(a_{\infty})v$ for any $v\in \mathit{F_{(n)}^{\ten \frac{1}{2}}}$, which since $\mathit{F_{(n)}^{\ten \frac{1}{2}}}=U(a_{\infty}^-)v_n$ proves that $\mathit{F_{(n)}^{\ten \frac{1}{2}}}$ is an irreducible highest weight module for $a_{\infty}$. The highest weights are calculated in Lemma \ref{highestweights}, and that proves  \propref{irreducibleFn}.
$ \square $

\medskip

\def\cprime{$'$}


 \end{document}